\documentclass[sigconf]{acmart}
\renewcommand\footnotetextcopyrightpermission[1]{} 

\usepackage{amsmath,amssymb}
\usepackage{tabularx}					
\usepackage[flushleft]{threeparttable}	
\usepackage{booktabs}

\usepackage{caption}
\usepackage{multirow}
\usepackage{graphicx}
\usepackage{array}
\usepackage{color}
\usepackage[linesnumbered,ruled,noend]{algorithm2e}
\usepackage{setspace}
\usepackage{pifont}
\usepackage{bbm}
\usepackage{subfigure}
\usepackage[normalem]{ulem}

\definecolor{orange}{rgb}{1, .36, .08}

\newcommand{\fakeparagraph}[1]{\vspace{1mm}\noindent\textbf{#1.}}

\setcopyright{acmcopyright}

\begin{document}

\title{
Toward Efficient Evaluation of Logic Encryption Schemes: Models and Metrics
} 

\author{Yinghua Hu$^{1}$ \quad Vivek V. Menon$^{2}$ \quad Andrew Schmidt$^{2}$ \quad 
 Joshua Monson$^{2}$ \\ Matthew French$^{2}$ \quad Pierluigi Nuzzo$^{1}$ \\ [-1ex]
\small$^{1}$ Department of Electrical and Computer Engineering, University of Southern California, Los Angeles, CA, USA, \{yinghuah, nuzzo\}@usc.edu\\[0ex]
\small$^{2}$ Information Sciences Institute, University of Southern California, Arlington, VA, USA, \{vivekv, aschmidt, jmonson, mfrench\}@isi.edu}

\begin{abstract}
Research in logic encryption over the last decade has resulted in various techniques to prevent different security threats such as Trojan insertion, intellectual property leakage, and reverse engineering. 
However, there is little agreement on a uniform set of metrics and models to efficiently assess the achieved security level and the trade-offs between security and overhead. 
This paper addresses the above challenges by relying on a general logic encryption model that can encompass all the existing techniques, and a uniform set of metrics
that can capture multiple, possibly conflicting, security concerns. We apply our modeling approach to four state-of-the-art encryption techniques, showing that it enables fast and accurate evaluation of design trade-offs, average prediction errors that are at least $2\times$ smaller than previous approaches, and 
the evaluation of compound encryption methods.\footnote{This report is an extended version of~\cite{hu2019models}.} 
\end{abstract}

\pagenumbering{gobble}
\maketitle

\section{\label{sec:intro}Introduction}

Integrated circuits (ICs) often represent the ultimate root of trust of modern computing systems. However, the decentralization of the IC design and manufacturing process over the years, involving multiple players in the supply chain, has increasingly raised the risk of hardware security threats from untrusted third parties.

\emph{Logic encryption} aims to counteract some of these threats by appropriately modifying the logic of a circuit, that is, by adding extra components and a set of key inputs such that the functionality of the circuit cannot be revealed until the correct value of the key is applied. 
Several logic encryption methods have been proposed over the the last decade to protect the designs from threats such as intellectual property (IP) piracy, reverse engineering, and hardware Trojan insertion (see, e.g., \cite{Roy2010Ending-piracy-o,Rajendran2015Fault-Analysis-, Yasin2016SARLock:-SAT-at, shamsi2019approximation,tehranipoor2010survey}). 
However, existing techniques are often tailored to specific attack models and security concerns, and rely on different metrics to evaluate their effectiveness. 
It is then difficult to quantify the security of different methods, rigorously evaluate the inherent trade-offs between different security concerns, and systematically contrast their strength with traditional area, delay, and power metrics. 

This paper introduces a formal modeling framework for the evaluation of logic encryption schemes and the exploration of the associated design space. We rely on a general functional model for logic encryption
that can encompass all the existing methods. 
Based on this general model, we make the following contributions:
\begin{itemize}
\item We define a set of metrics that can formally capture multiple, possibly conflicting, security concerns that are key to the design of logic encryption schemes, such as functional corruptibility and resilience to different attacks, thus providing a common ground to compare different methods. 
  \item We develop compact models to efficiently quantify the quality and resilience of four methods, including state-of-the-art logic encryption techniques, and enable trade-off evaluation between different security concerns. 
\end{itemize}
Simulation results on a set of ISCAS benchmark circuits show the effectiveness of our modeling framework for fast and accurate evaluation of the design trade-offs. 
Our models produce conservative estimates of resilience with average prediction errors that are at least twice as small as previous approaches and, in some cases, improve by two orders of magnitude.
Finally, our approach can provide quantitative support to inform system-level decisions across multiple logic encryption strategies as well as the implementation of compound strategies, which can be necessary for providing high levels of protection against different threats with limited overhead. 

The rest of the paper is organized as follows. Section~\ref{sec:background} introduces  background concepts on logic encryption and recent efforts toward the systematic analysis of their security properties. Section~\ref{sec:theory} presents the general functional model for combinational logic encryption 
and defines four security-driven evaluation metrics. 
Section~\ref{sec:obf_lib} applies the proposed model and metrics to the analysis of the security properties of four encryption techniques. 
Our analysis is validated in Section~\ref{sec:experiment} and compared with state-of-the-art characterizations of the existing techniques. 
Finally, Section~\ref{sec:conclusion} concludes the paper.

\section{Background and Related Work}\label{sec:background}

Logic encryption techniques 
have originally focused mostly on a subset of security concerns, and lacked 
methods to systematically quantify the level of protection against different (and potentially unknown) hardware attacks. A class of methods, such as \emph{fault analysis-based logic locking} (\texttt{FLL})~\cite{Rajendran2015Fault-Analysis-}, mostly focuses on providing high output \emph{error rates} 
when applying a wrong key, for example, by appropriately inserting key-controlled XOR and XNOR gates in the circuit netlist. 
Another class of techniques, based on one-point functions, such as \texttt{SARLock}~\cite{Yasin2016SARLock:-SAT-at}, 
aims, instead, to provide resilience to \emph{SAT-based attacks}, a category of attacks using satisfiability (SAT) solving to efficiently prune the space of possible keys~\cite{Subramanyan2015Evaluating-the-}. These methods require an exponential number of SAT-attack iterations in the size of the key to unlock the circuit, but tend to
expose a close approximation of the correct circuit function. Efforts toward a comprehensive encryption framework 
have only started to appear. 

\emph{Stripped functionality logic locking} (\texttt{SFLL})~\cite{Yasin2017Provably-Secure} has been recently proposed as a scheme for provably secure encryption with respect to a broad set of quantifiable security concerns, including error rate, resilience to SAT attacks, and resilience to removal attacks, aiming to remove the encryption logic from the circuit. However, 
while 
the average number of SAT-attack iterations is shown to grow exponentially with the key size, the worst-case SAT-attack duration, as discussed in Sec.~\ref{sec:obf_lib}, can become unacceptably low, which calls for mechanisms to explore the combination of concepts from \texttt{SFLL} with other schemes. 

Zhou~\cite{zhou2017humble} provides a theoretical analysis of the contention between error rate and SAT-attack resilience in logic encryption, drawing from concepts in learning theory~\cite{valiant1984theory}. 
Along the same direction, Shamsi~\textit{et al.}~\cite{shamsi2019approximation} develop \emph{diversified tree logic} (\texttt{DTL}) as a scheme capable of increasing the error rate of SAT-resilient protection schemes in a tunable manner. 
A recent effort~\cite{shamsi2019locking} adopts a game-theoretic approach to formalize notions of secrecy and resilience that account for the impact of learnability of the encrypted function and information leakage from the circuit structure. 

While our approach builds on previous analyses~\cite{zhou2017humble, shamsi2019approximation}, it is complementary, as it focuses on 
models and metrics that enable fast and accurate evaluation across multiple encryption techniques and security concerns,  
eventually raising the level of abstraction at which security-related design decisions can be made. 
%
We distinguish between logic encryption, which augments the circuit function via additional components and key bits, and obfuscation~\cite{barak2001possibility,goldwasser2007best}, which is concerned with hiding the function of a circuit or program (without altering it) to make it unintelligible from its structure. In this paper, we focus on the functional aspects of logic encryption, and leave the modeling of its interactions with obfuscation 
for future work. 

\section{Logic Encryption: Models and Metrics}
\label{sec:theory}



We denote by $|S|$ the cardinality of a set $S$. We represent a combinational logic circuit with primary input (PI) ports $I$ and primary output (PO) ports $O$ by its Boolean function $f: \mathbb{B}^n \rightarrow \mathbb{B}^m$, where $n=|I|$ and $m=|O|$, and its netlist, modeled as a labelled directed graph $G$.
Both $f$ and $G$ may be parameterized by a set of  configuration parameters $P$, with values in $\mathcal{P}$, related to both the circuit function and implementation. 
Given a function $f$, logic encryption creates a new function $f': \mathbb{B}^n \times \mathbb{B}^l \rightarrow \mathbb{B}^m$, 
where $l=|K|$ and $K$ is the set of key input ports added to the netlist. There exists $k^* \in \mathbb{B}^l$ such that $\forall i \in \mathbb{B}^n, f(i)\equiv f'(i, k^*)$. We call $k^*$ the correct key. We wish to express $f'$ as a function of $f$ and the encryption logic. 

\subsection{A General Functional Model}

We build on the recent literature~\cite{zhou2017humble,yasin2017ttlock,Yasin2017Provably-Secure} to define a general model, capable of representing the behavior of all the existing logic encryption schemes, as shown in Fig.~\ref{fig:general_locking}. 
The function $g(i, k)$ maps an input and key value to a \emph{flip} signal, which is combined with the output of $f(i)$ via a XOR gate to produce the encrypted PO.  
The value of the PO is inverted when the \textit{flip} signal is one. 
We assume that $g$ is parameterized by a set $Q$ of configuration parameters, with values in $\mathcal{Q}$ related to a specific encryption technique. 
%
%


\begin{figure}[t]
	\centering
	\includegraphics[width=0.75\columnwidth]{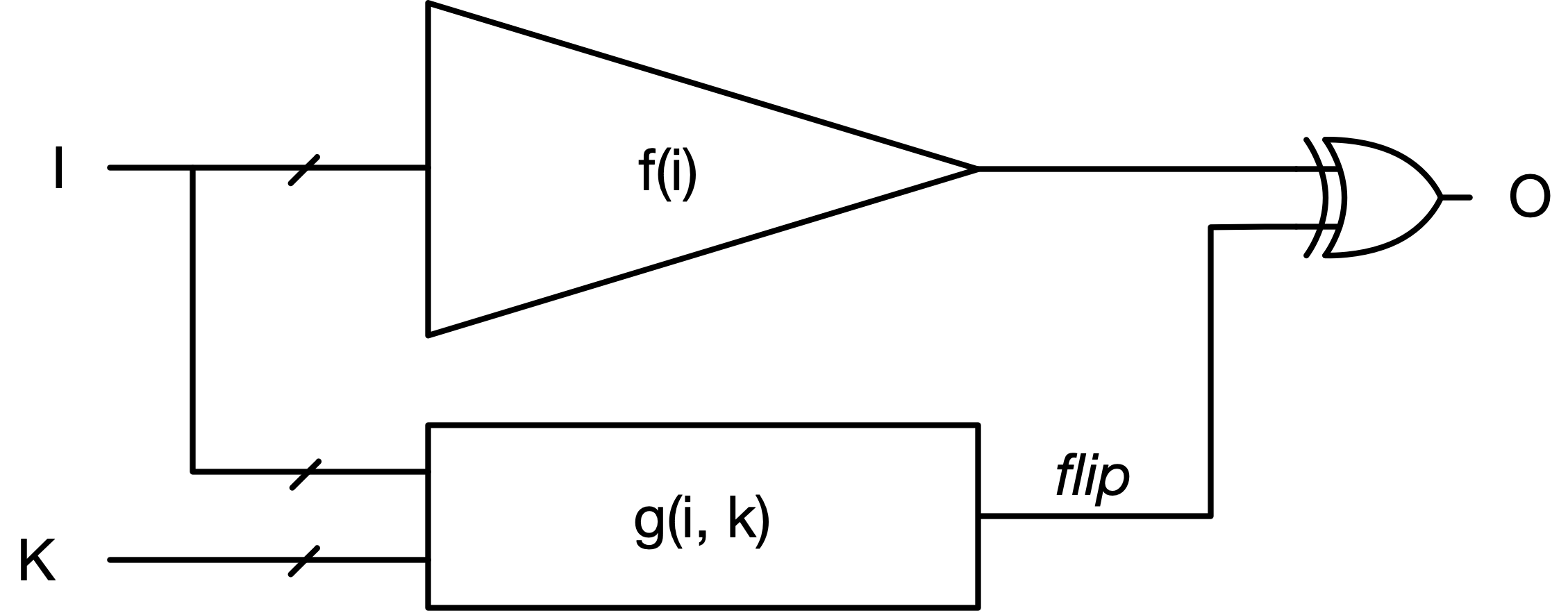}
	\caption{General functional model for logic encryption. 
	}
	\vspace{-10pt}
	\label{fig:general_locking}
\end{figure}

\subsection{Security-Driven Metrics}

We can describe how the circuit output is affected by logic encryption via an \emph{error table}, such as the ones shown in Tab.~\ref{tab:truth_table_example}. 
Based on the general functional model above and the associated error tables, we define a set of security-driven metrics that capture the quality and resilience of encryption.

\fakeparagraph{Functional Corruptibility} Functional corruptibility quantifies the amount of output error induced by logic encryption to protect the circuit function. 
Consistently with the literature~\cite{shamsi2019approximation}, we define the \emph{functional corruptibility} $E_{FC}$ as the ratio between the number of corrupted output values and the total number of primary input and key configurations (the entries in the error table), i.e.,
$$ 
E_{FC} = \frac{1}{2^{n+l}} \sum_{i\in \mathbb{B}^n}\sum_{k\in \mathbb{B}^l}\mathbbm{1}(f(i)\neq f'(i, k)),$$
where $\mathbbm{1}(A)$ is the indicator function, evaluating to $1$ if and only if  event $A$ occurs. 

\begin{table}[t]
\centering
\caption{Error tables with $n=l=3$ ({\color{red}\ding{54}} and {\ding{52}} mark incorrect and correct output values, respectively).}
\vspace{-15pt}
\subfigure[\texttt{SARLock}]{
\scalebox{0.56}{
\centering
\begin{tabular}{|c|c|c|c|c|c|c|c|c|}
\hline
  & \textbf{K0} & \textbf{K1} & \textbf{K2} & \textbf{K3} & \textbf{K4} & \textbf{K5} & \textbf{K6} & \textbf{K7} \\ \hline
\textbf{I0} & {\color{red}\ding{54}} & \ding{52} & \ding{52} & \ding{52} & \ding{52} & \ding{52} & \ding{52} & \ding{52} \\ \hline
\textbf{I1} & \ding{52} & \ding{52} & {\color{red}\ding{54}} & \ding{52} & \ding{52}  & \ding{52}  & \ding{52}  & \ding{52} \\ \hline
\textbf{I2} &  \ding{52} & {\color{red}\ding{54}} & \ding{52} &  \ding{52} & \ding{52} & \ding{52} & \ding{52}  & \ding{52} \\ \hline
\textbf{I3} & \ding{52} & \ding{52}  & \ding{52} & {\color{red}\ding{54}} & \ding{52} & \ding{52}  & \ding{52} & \ding{52}  \\ \hline
\textbf{I4} & \ding{52}  & \ding{52}  & \ding{52}  & \ding{52}  & {\color{red}\ding{54}} & \ding{52} & \ding{52}  & \ding{52}  \\ \hline
\textbf{I5} & \ding{52} &  \ding{52} & \ding{52}  & \ding{52} & \ding{52} &  \ding{52} & \ding{52} & \ding{52} \\ \hline
\textbf{I6} & \ding{52} &\ding{52}  & \ding{52}  & \ding{52} & \ding{52} & {\color{red}\ding{54}}  & \ding{52} & \ding{52} \\ \hline
\textbf{I7} & \ding{52} & \ding{52} & \ding{52} & \ding{52}  & \ding{52}  & \ding{52} &  \ding{52} & {\color{red}\ding{54}} \\ \hline
\end{tabular}}
}
\subfigure[\texttt{SFLL-HD0}]{
\centering
\scalebox{0.56}{
\begin{tabular}{|c|c|c|c|c|c|c|c|c|}
\hline
  & \textbf{K0} & \textbf{K1} & \textbf{K2} & \textbf{K3} & \textbf{K4} & \textbf{K5} & \textbf{K6} & \textbf{K7} \\ \hline
\textbf{I0} & {\color{red}\ding{54}} & \ding{52} & \ding{52} & \ding{52} & \ding{52} & \ding{52} & \ding{52} & \ding{52} \\ \hline
\textbf{I1} &  \ding{52} & {\color{red}\ding{54}} &\ding{52}& \ding{52} & \ding{52}  & \ding{52}  & \ding{52}  & \ding{52} \\ \hline
\textbf{I2} &  \ding{52} & \ding{52}& {\color{red}\ding{54}}  &  \ding{52} & \ding{52} & \ding{52} & \ding{52}  & \ding{52} \\ \hline
\textbf{I3} & \ding{52} & \ding{52}  & \ding{52} & {\color{red}\ding{54}} & \ding{52} & \ding{52}  & \ding{52} & \ding{52}  \\ \hline
\textbf{I4} & \ding{52}  & \ding{52}  & \ding{52}  & \ding{52}  & {\color{red}\ding{54}} & \ding{52} & \ding{52}  & \ding{52}  \\ \hline
\textbf{I5} & \ding{52} &  \ding{52} & \ding{52}  & \ding{52} & \ding{52} &  {\color{red}\ding{54}} & \ding{52} & \ding{52} \\ \hline
\textbf{I6} & {\color{red}\ding{54}} &{\color{red}\ding{54}}  & {\color{red}\ding{54}}  & {\color{red}\ding{54}} & {\color{red}\ding{54}} & {\color{red}\ding{54}}  & \ding{52} & {\color{red}\ding{54}} \\ \hline
\textbf{I7} & \ding{52} & \ding{52} & \ding{52} & \ding{52}  & \ding{52}  & \ding{52} &  \ding{52} & {\color{red}\ding{54}} \\ \hline
\end{tabular}}
}
\vspace{-15pt}
\label{tab:truth_table_example}
\end{table}

\begin{figure*}[t]
	\centering
	\vspace{-20pt}
	\subfigure[]{
  \centering
  \includegraphics[width=0.5\columnwidth]{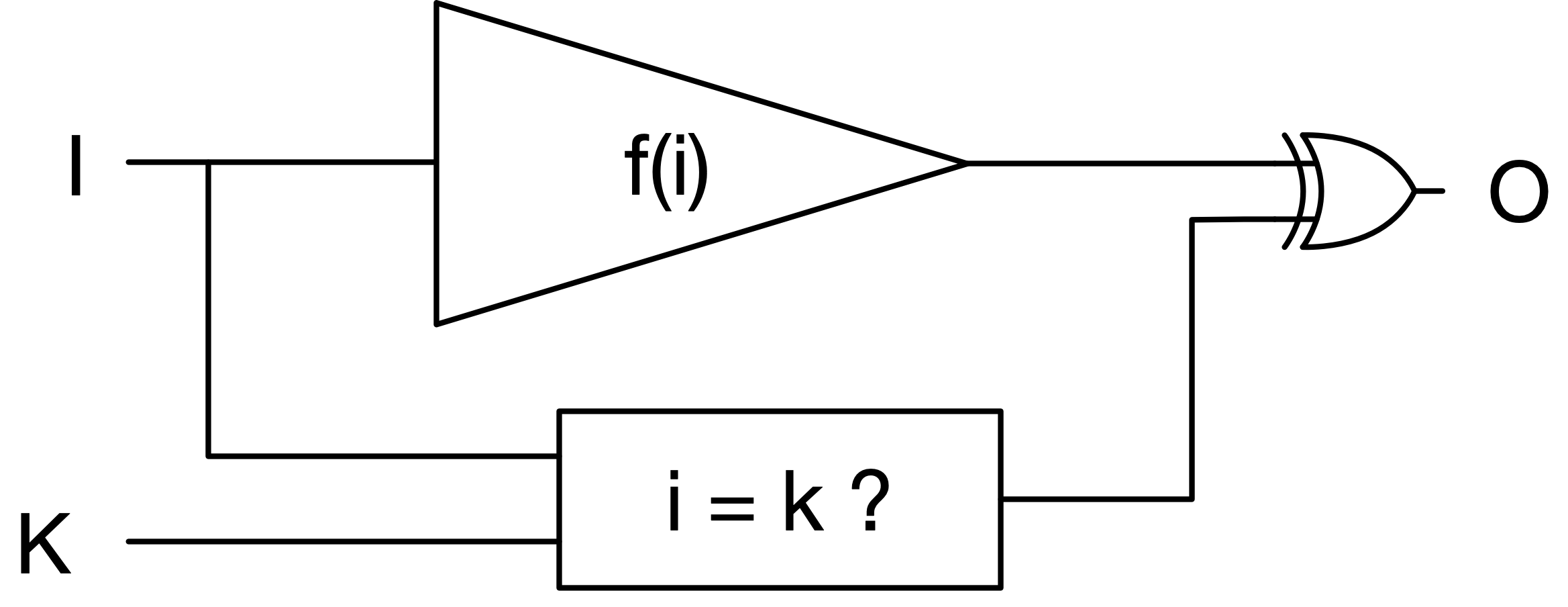}
  }%
  \subfigure[]{
  \centering
  \includegraphics[width=0.5\columnwidth]{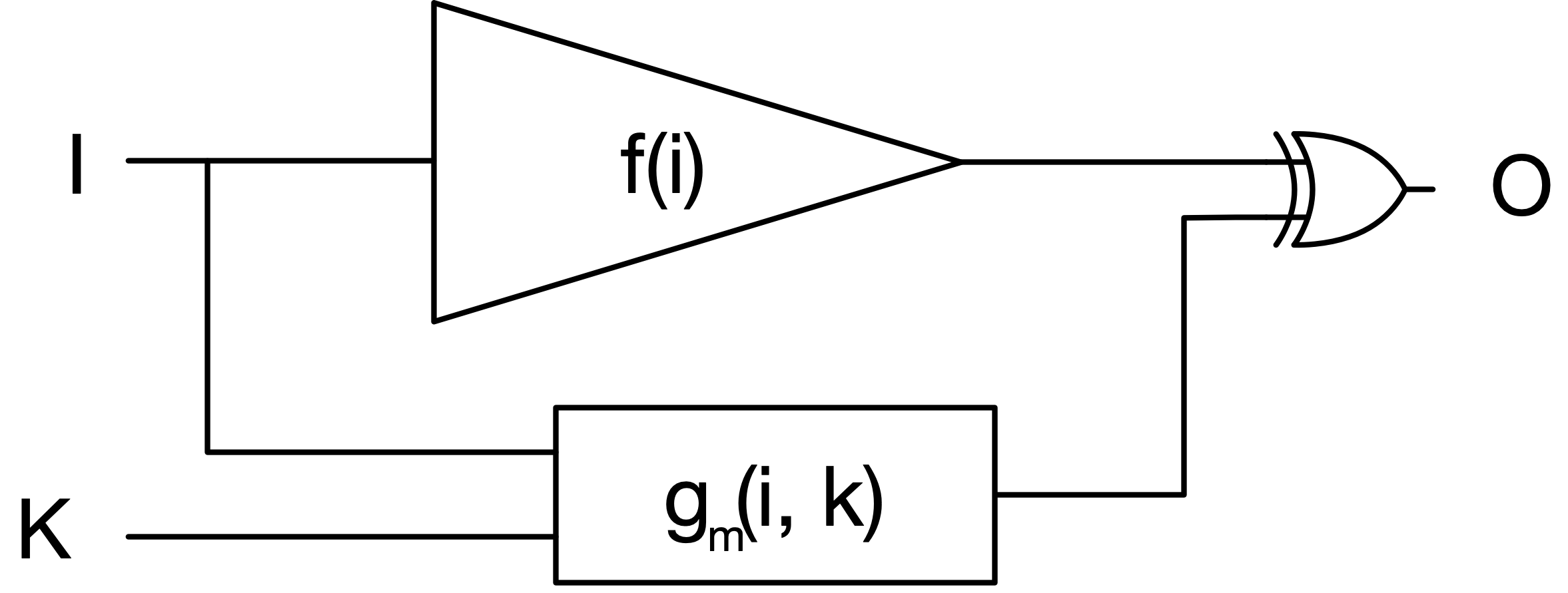}
  }%
  \subfigure[]{
  \centering
  \includegraphics[width=0.5\columnwidth]{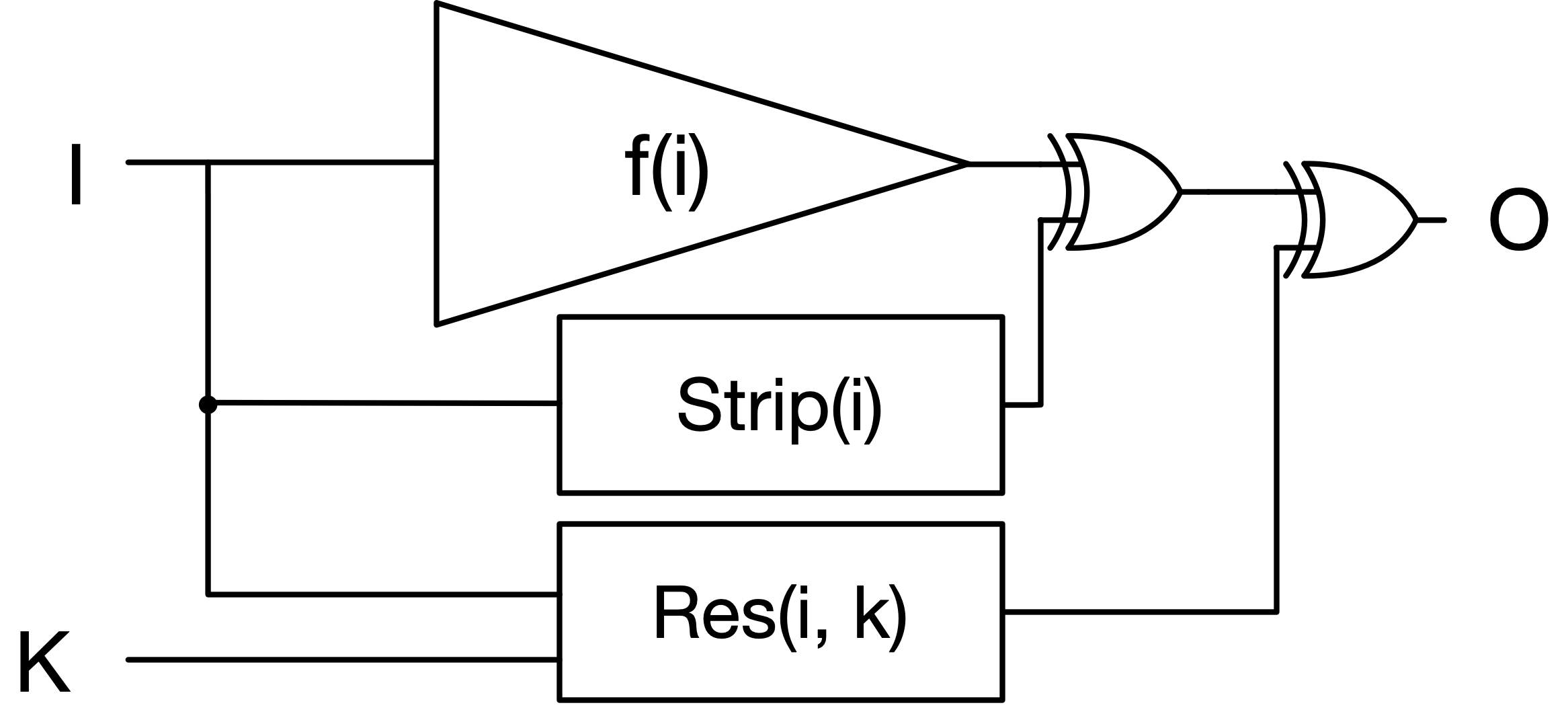}
  }%
  \subfigure[]{
  \centering
  \includegraphics[width=0.5\columnwidth]{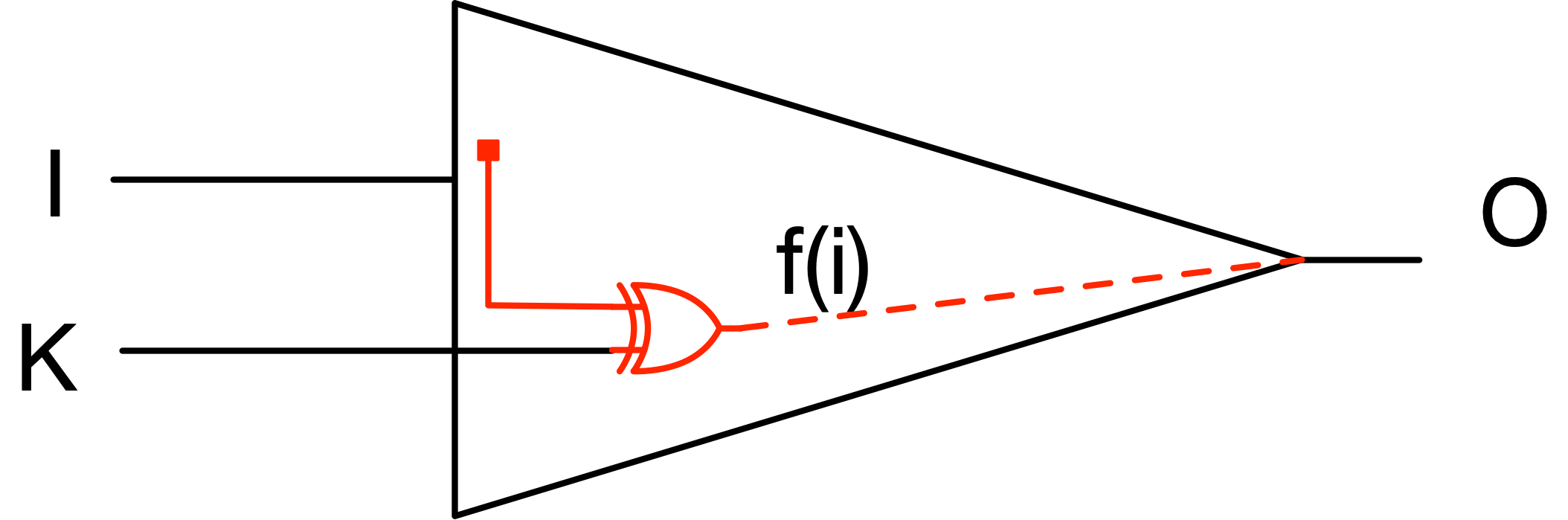}
  }%
	\vspace{-10pt}
	\caption{Functional model for (a) \texttt{SARLock}, (b) \texttt{DTL}, (c) \texttt{SFLL}, and (d) \texttt{FLL}. 
	}
	\label{fig:general_locking_decomposition}
	\vspace{-10pt}
\end{figure*}

\fakeparagraph{SAT Attack Resilience ($t_{SAT}$)} 
A SAT attack~\cite{Subramanyan2015Evaluating-the-} assumes that 
the attacker has access to the encrypted netlist and an operational (deobfuscated) circuit, used as an oracle, to query for correct input/output pairs. The goal is to reconstruct the exact circuit function by retrieving a correct key. At 
each iteration, the attack solves a SAT problem to search for a \emph{distinguishing input pattern} (DIP), that is, an input pattern $i$ that provides different output values for different keys, i.e., such that 
$\exists \ k_1\neq k_2, f'(i, k_1) \neq f'(i, k_2)$. The attack then queries the oracle to find the correct output $f(i)$ and incorporate this information in
the original SAT formula to constrain the search space for the following iteration. Therefore, all the keys leading to an incorrect output value for the current DIP will be pruned out of the search. Once the SAT solver cannot find a new DIP, the SAT attack terminates marking the remaining keys as correct. 

Consistently with the literature~\cite{Yasin2016SARLock:-SAT-at, Yasin2017Provably-Secure}, we quantify the hardness of this attack using the number of SAT queries, hence the number of DIPs, required to obtain the circuit function. 
Predicting this number in closed form is challenging, since it relates to solving a combinatorial search problem, in which the search space generally depends on the circuit properties and the search heuristics on the specific solver or algorithm adopted. Current approaches~\cite{Yasin2017Provably-Secure} adopt probabilistic models, where the expected number of DIPs is computed under the assumption that the input patterns are searched according to a uniform distribution. We adopt, instead, a worst-case conservative model and use the minimum number of DIPs to quantify the guarantees of an encryption technique in terms of SAT-attack resilience.
The duration of the attack also depends on the circuit size and structure, since they affect the runtime of each SAT query. 
In this paper, we regard the runtime of each SAT query as a constant and leave a more accurate 
modeling of the duration of the attack for future work.

\fakeparagraph{Approximate SAT-Attack Resilience ($E_{APP}$)}
Approximate SAT attacks, such as  \emph{AppSAT}~\cite{shamsi2017appsat} and \emph{Double-DIP}~\cite{shen2017double},  perform a variant of a SAT attack but terminate earlier, when the error rate at the PO is ``low enough,'' providing a sufficient approximation of the circuit function. In this paper, we take a worst-case approach by assuming that an approximate SAT attack terminates in negligible time, and define the approximate SAT-attack resilience ($E_{APP}$) as the minimum residual error rate that can be obtained with an incorrect key (different than $k^*$), i.e., 
$$ E_{APP} = \min_{k \in \mathbb{B}^l \setminus \{k^*\}}\frac{\epsilon_k}{2^n},$$
where $\epsilon_k$ is the number of incorrect output values for key input $k$. 

\fakeparagraph{Removal Attack Resilience ($E_{REM}$)} A removal attack consists in  directly removing all the added encryption logic to unlock a circuit, e.g., by bypassing the flip signal~\cite{Yasin2017Provably-Secure} or  the key-controlled XOR/XNOR   gates~\cite{chakraborty2018sail}. 
We make the worst-case assumption that all the key-related components can be removed from the encrypted netlist in negligible time. We then define the resilience metric as the ratio of input patterns that are still protected after removal, i.e., 
$$ E_{REM} = \frac{\sum_{i\in \mathbb{B}^n}\mathbbm{1}(f_{REM}(i)\neq f(i))}{2^n}$$
where $f_{REM}(.)$ is the Boolean function obtained after removing all the key-related components. 


\section{Encryption Methods}\label{sec:obf_lib}

We apply the general model and metrics in Sec.~\ref{sec:theory} to four logic encryption techniques, namely, \texttt{SARLock}, \texttt{SFLL}, \texttt{DTL}, and \texttt{FLL}, 
showing that it encompasses existing methods, including state-of-the-art techniques. 
Tab.~\ref{tab:security_model_example} summarizes the security models with respect to the four security metrics described in Sec.~\ref{sec:theory}. The models, including proofs for our results, are discussed in detail below. 

\begin{table}[t]
\centering
\caption{Security metrics for four logic encryption techniques ($n=|I|$, $m=|O|$, and $l=|K|$). 
}
\label{tab:chance_cons_formu}
\resizebox{\columnwidth}{!}{
\scalebox{1}{
\begin{tabular}{c|c|c|c|c}
\textbf{ }
& $\mathbf{E_{FC}}$ 
& $\mathbf{t_{SAT}}$ 
& $\mathbf{E_{APP}}$ 
& $\mathbf{E_{REM}}$ \\
\hline
\rule[-0.75em]{0pt}{1em}\rule[0.25em]{0pt}{1em}
\textbf{SARLock} & $\approx\frac{1}{2^{l}}$ & $\min \left\{2^{l}, 2^{n}\right\}$  & $\frac{1}{2^{l}}$   & 0    \\ 
\hline
\rule[-1em]{0pt}{1.6em}\rule[0.5em]{0pt}{1.6em}
\textbf{DTL}  & $\approx\frac{\left(2\left(2^{2^{L}}-1 \right)\right)^N}{2^{l}}$ & $\min \left\{\frac{2^{l}}{(2(2^{2^{L}}-1))^N}, 2^{n}\right\}$  & $\frac{\left(2\left(2^{2^{L}}-1 \right)\right)^N}{2^{l}}$   & 0    \\ 
\hline
\rule[-0.75em]{0pt}{1.25em}\rule[0.5em]{0pt}{1.5em}
\textbf{SFLL} & $\frac{\binom{l}{h}\left[2^{l} - \binom{l}{h}\right]}{2^{2l-1}}$ & $<< \exp(l)$  & $\frac{2\left[\binom{l}{h}-2\binom{l-2}{h-1}\right]}{ 2^{n}}$   & $\binom{l}{h}/2^{l}$    \\ 
\hline
\textbf{FLL}  & $[0.3, 0.5]$ & $<< \exp(l)$ & $<< E_{FC}$ & $0$    \\ 
\end{tabular}}
}
\vspace{-10pt}
\label{tab:security_model_example}
\end{table}

\fakeparagraph{SARLock}
\texttt{SARLock} combines the output of the original circuit $f(i)$ with the one-point function $\mathbbm{1}(i=k)$. 
It can then be mapped to the general functional model where $g = \mathbbm{1}(i=k)$, as shown in Fig.~\ref{fig:general_locking_decomposition}a. 
The parameter set $Q_{\texttt{SARLock}}$ includes the key size $l=|K|$. Consistently with previous work~\cite{Yasin2016SARLock:-SAT-at}, we derive the closed form expressions in Tab.~\ref{tab:security_model_example} as stated by the following result. 

\begin{theorem}\label{theorem:sarlock}
For a circuit encrypted with \texttt{SARLock}, let $l$ and $n$ be the key size and the primary input size, respectively. Let $E_{FC}$ be the functional corruptibility,  $t_{SAT}$ the SAT-attack resilience, $E_{APP}$ the approximate SAT-attack resilience, and  $E_{REM}$ the removal attack resilience. Then, the following equations hold: $E_{FC} = \frac{1}{2^{l}}$, $t_{SAT} = \min\left\{2^l, 2^n\right\}$,  $E_{APP} = \frac{1}{2^{l}}$, and $E_{REM} = 0$. 
\end{theorem}

\begin{proof}
We observe that the key size $l$ can be at most equal to the primary input size $n$, i.e., $l \le n$ holds. For an incorrect key $k$, the output is corrupted only when the input $i$ is equal to $k$. Therefore, the number of corrupted output patterns is $2^{n-l}$ for each incorrect key. Because there are $2^l -1$ incorrect keys, we can compute $E_{FC}$ and $E_{APP}$ as follows: %
\begin{equation}
    \begin{aligned}
        E_{FC} &= \frac{2^{n-l}\cdot (2^{l}-1)}{2^{n}\cdot 2^{l}} \approx \frac{1}{2^l},
    \end{aligned}
\end{equation}
\begin{equation}
    \begin{aligned}
        E_{APP} &= \frac{2^{n-l}}{2^n} = \frac{1}{2^l}.
    \end{aligned}
\end{equation}

By definition of \texttt{SARLock}, each input pattern can only exclude one incorrect key at each iteration of a SAT attack (see, for example, the error table in Tab.~\ref{tab:truth_table_example}). Because there are $2^l - 1$ incorrect keys to exclude, and the number of SAT attack iterations is bounded above by $2^n$, the total number of primary input patterns, 
we can compute $t_{SAT}$ as follows:
\begin{equation}
    \begin{aligned}
        t_{SAT} &= \min\left\{2^l, 2^n\right\}. 
    \end{aligned}
\end{equation}

Finally, by the definition of removal attack resilience, once the flip signal of the one-point function is recognized and bypassed, the original functionality of the circuit is fully restored, leading to $E_{REM}=0$. 
\end{proof}

The use of a one-point function, especially when the key size $l$ is very large, results in very low functional corruptibility $E_{FC}$ but exponential SAT-attack resilience $t_{SAT}$, as stated by Theorem~\ref{theorem:sarlock}. A moderately high $E_{FC}$ can still be achieved, but this happens with small key sizes. For example, $E_{FC} =  0.25$ can be achieved for $l=2$. 

 

\fakeparagraph{Diversified Tree Logic (DTL)}
The one-point functions used in \texttt{SARLock} or \texttt{Anti-SAT}~\cite{Xie2018Anti-SAT:-Mitig} are based on AND-tree structures. An example of a four-input AND-tree is shown in Fig.~\ref{fig:dtl_detail}.
\begin{figure}[t]
	\centering
	\includegraphics[width=0.55\columnwidth]{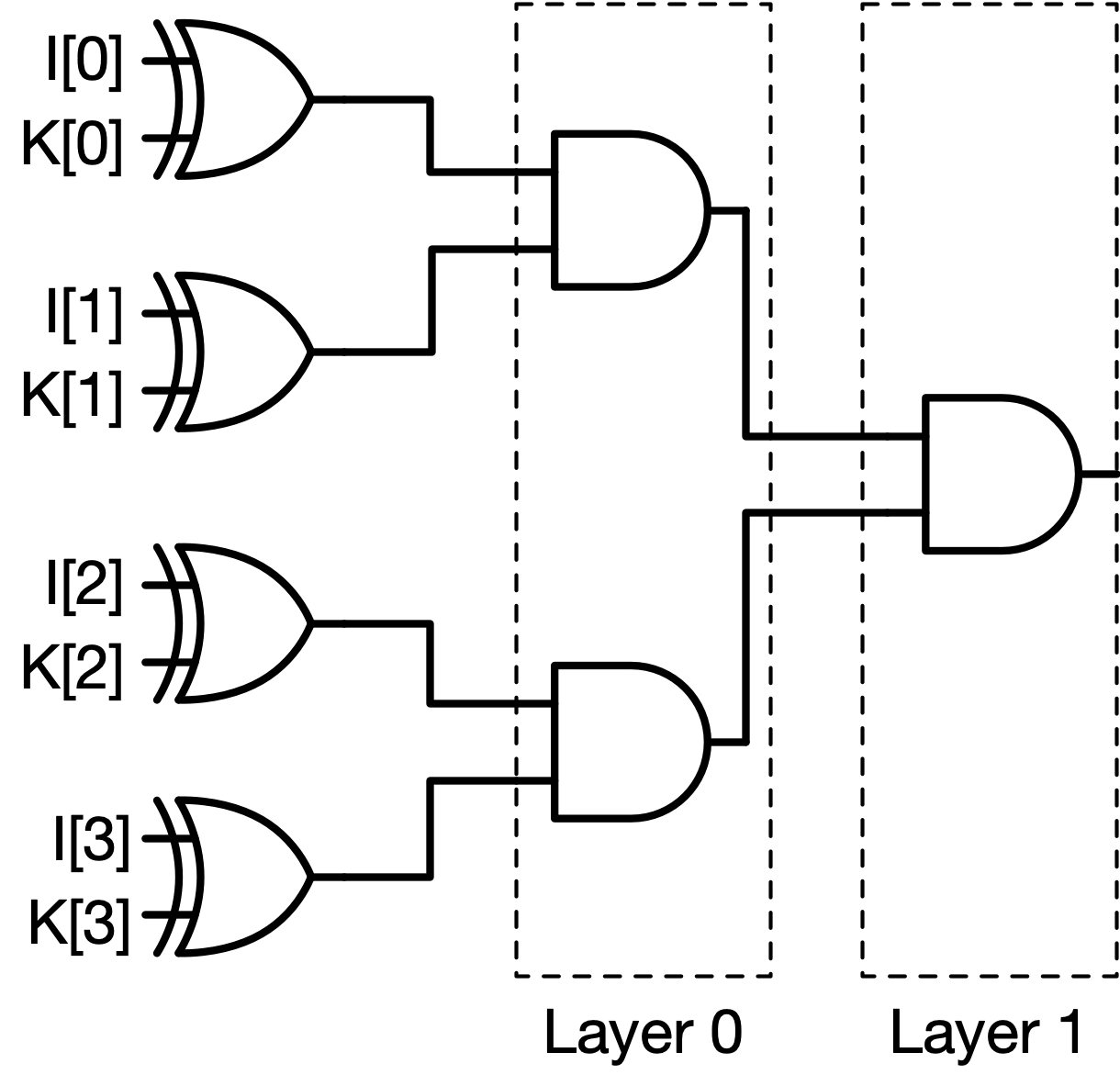}
	\caption{A four-input AND-tree structure.
	}
	\label{fig:dtl_detail}
\end{figure}
\texttt{DTL} borrows such structures from \texttt{SARLock} or \texttt{Anti-SAT} and appropriately replaces some of the AND gates with another type of gate, i.e., XOR, OR, or NAND, to obtain a multi-point function $g_m(i.k)$, as shown in Fig.~\ref{fig:general_locking_decomposition}b. 
The parameter set $Q_{\texttt{DTL}}$ includes: (1) the key size $|K|$; (2) the type of point-function $T$, e.g., $T \in \{\texttt{SARLock},\texttt{Anti-SAT}\}$; (3) the replacement tuple $(X, L, N)$, where $X \in \{\texttt{XOR}, \texttt{OR}, \texttt{NAND}\}$ is a gate type for the replacement, and $L$ and $N$ denote the layer and number of gates selected for replacement, respectively, with $0 \leq L \leq \lceil \log_2(|K|)\rceil -1$, $0\leq N \leq 2^{\lceil \log_2(|K|)\rceil-L-1}$, and $|K| \geq 2$. 
\texttt{DTL} then modifies the AND-tree by replacing $N$ gates from layer $L$ with gates of type $X$. 
Tab.~\ref{tab:security_model_example} shows the expressions obtained when $T=$ \texttt{SARLock}, and $X=$ \texttt{XOR}, as stated by the following theorem. 

\begin{theorem}\label{DTL_theorem}
For a circuit encrypted with \texttt{DTL} of type $T=$ \texttt{SARLock} and replacement gate $X=$ \texttt{XOR}, let $l$ and $n$ be the key size and the primary input size, respectively. Let $E_{FC}$ be the functional corruptibility,  $t_{SAT}$ the SAT-attack resilience, $E_{APP}$ the approximate SAT-attack resilience, and  $E_{REM}$ the removal attack resilience. Then, the following equations hold: 
$E_{FC}=\frac{\left(2\left(2^{2^{L}}-1 \right)\right)^N}{2^{l}}$, $t_{SAT}=\min \left\{\frac{2^{l}}{(2(2^{2^{L}}-1))^N}, 2^{n}\right\}$, $E_{APP}=\frac{\left(2\left(2^{2^{L}}-1 \right)\right)^N}{2^{l}}$, and $E_{REM}=0$. 
\end{theorem}

\begin{proof}
According to the analysis 
by Shamsi~et~al.~\cite{shamsi2019approximation}, replacing an AND gate in the first layer ($L=0$) of a \texttt{SARLock} block with a XOR gate changes the onset size of the one-point function from 1 to 2. More generally, replacing $N$ AND gates in layer $L$ changes the onset size to $\left(2\left(2^{2^{L}}-1 \right)\right)^N$. As a result, for each incorrect key (an incorrect column in the error table), there are $2^{n-l}\cdot \left(2\left(2^{2^{L}}-1 \right)\right)^N$ incorrect output patterns. $E_{FC}$ and $E_{APP}$ can then be computed as 
\begin{equation}
    \begin{aligned}
        E_{FC} &= \frac{(2^{l}-1)\cdot 2^{n-l}\cdot \left(2\left(2^{2^{L}}-1 \right)\right)^N}{2^{n}\cdot 2^{l}} \approx \frac{\left(2\left(2^{2^{L}}-1 \right)\right)^N}{2^l},
    \end{aligned}
\end{equation}
\begin{equation}
    \begin{aligned}
        E_{APP} &= \frac{2^{n-l}\cdot \left(2\left(2^{2^{L}}-1 \right)\right)^N}{2^n} = \frac{\left(2\left(2^{2^{L}}-1 \right)\right)^N}{2^l}.
    \end{aligned}
\end{equation}

From the analysis by  Shamsi~et~al.~\cite{shamsi2019approximation}, the number of SAT queries needed is $\frac{2^{l}}{(2(2^{2^{L}}-1))^N}$. Because of the upper bound due to the maximum number of input patters $2^n$, we obtain: 
\begin{equation}
    \begin{aligned}
        t_{SAT} &= \min \left\{\frac{2^{l}}{(2(2^{2^{L}}-1))^N}, 2^{n}\right\}.
    \end{aligned}
\end{equation}

Finally, the flip signal of \texttt{DTL} can be recognized and bypassed, which returns the full functionality of the original circuit and leads to 
$        E_{REM}=0$.
\end{proof}

\noindent In \texttt{DTL}, the error table has the same number of errors in each column, except for the correct key column, which makes $E_{FC}$ equal to $E_{APP}$. 
$N$
can be tuned to increase $E_{APP}$ and $E_{FC}$ while $t_{SAT}$ decreases. 
Analogous results as in Theorem~\ref{DTL_theorem} 
can be derived for other configurations of $T$ and $X$~\cite{shamsi2019approximation}. 
Specifically, for all $X$ and $T$, we obtain: 
\begin{equation}\label{eqn:approx_relation_dtl_fc}
    \begin{aligned}
        E_{FC} = E_{APP} = O\left(2^{N\cdot 2^L-l}\right)
    \end{aligned}
\end{equation}
\begin{equation}\label{eqn:approx_relation_dtl_sat}
    \begin{aligned}
        t_{SAT} &= O \left(\min \left\{2^{l - N\cdot 2^L}, 2^n\right\} \right).
    \end{aligned}
\end{equation}
Based on the expressions above, the maximum $E_{FC}$ or $E_{APP}$ can be achieved when all the AND gates are replaced in a given layer. The approximate security levels in this scenario  show the following behavior: 
\begin{equation}
    \begin{aligned}
        E_{FC, max} = E_{APP, max} = O \left( 2^{-\frac{l}{2}} \right)
    \end{aligned}
\end{equation}
\begin{equation}
    \begin{aligned}
        t_{SAT} &= O \left( \min \left\{2^{\frac{l}{2}}, 2^n\right\} \right).
    \end{aligned}
\end{equation}

\fakeparagraph{SFLL} 
Fig.~\ref{fig:general_locking_decomposition}c shows the schematic of \texttt{SFLL}, where  the value of the primary output is given by $f(i)\oplus Strip(i) \oplus Res(i,k)$. Both the stripping circuit $Strip(i)$ and the restore circuit $Res(i,k)$ are point functions. The stripping block corrupts part of the original function, while the restore unit restores the correct value upon applying the correct key. \texttt{SFLL} can be mapped to the general functional model with $g(i,k) = Strip(i) \oplus Res(i,k)$. 
We focus on \texttt{SFLL-HD} where $Res(i,k)$ is a Hamming distance comparator. The parameter set $Q_{\texttt{SFLL}}$ includes $|K|$ and $h$, representing the key size and the HD parameter for the HD comparator, respectively. The comparator output will evaluate to one if and only if the HD between its inputs is $h$. The key size must be at most equal to the number of PI ports in the fan-in cone of the protected PO port, i.e., $0 \leq |K| \leq |I|$, while $h$ is at most equal to $|K|$. 

\texttt{SFLL} is the only technique  in Tab.~\ref{tab:security_model_example} that is resilient to removal attacks. 
In fact, at the implementation level, $Strip(i)$ can be merged with $f(i)$ to form a monolithic block (e.g., via a re-synthesis step or modification of internal signals of the original circuit) and, therefore, it becomes hard to remove. On the other hand, unlike \texttt{SARLock}, it does not guarantee exponential SAT-attack resilience. For example, as shown in Tab.~\ref{tab:truth_table_example}b for  $h=0$, selecting  one input pattern, such as $I6$, is enough to prune out all the incorrect keys and unlock the circuit after one SAT-attack iteration. Previous work~\cite{Yasin2017Provably-Secure} proposes a probabilistic model in terms of \emph{expected number of DIPs}, based on the following assumptions: (i) SAT solvers select input patterns with a uniform distribution; (ii) the probability of terminating a SAT attack is equal to the probability of finding one protected input pattern, i.e., finding one protected input pattern is enough to prune out all the incorrect keys and terminate a SAT attack. Based on this model, the average SAT resilience of \texttt{SFLL} is shown to increase exponentially with $l$. 

We find that the existing probabilistic models tend to become inaccurate when $h$ is different than $0$ or $l$, since, in these configurations, one protected input pattern is generally not enough to terminate a SAT attack. Moreover, these models tend to ignore the heuristics adopted by state-of-the-art SAT solver to accelerate the search. In this paper, we adopt, instead, a conservative metric in terms of \emph{minimum number of DIPs}.
The following results 
states the hardness of finding the minimum number of DIPs. 

\begin{theorem}\label{theorem:min_cover}
Given an encrypted Boolean function of 
the primary and key inputs, computing the minimum number of distinguishing input patterns (DIPs) for a SAT attack can be reduced to a min-set-cover problem, which is NP-hard~\cite{korte2012combinatorial}. 
\end{theorem}

\begin{proof}
Given 
the error table associated with an encrypted Boolean function $f'$, 
let $U$ be the set of all the incorrect keys, i.e.,
$$U = \{ k | \exists \ i \in \mathbb{B}^{|I|}, f(i)\neq f'(i,k)\}.$$
For each input pattern $i$, let $S_i$ be the set of the incorrect keys that can be eliminated by a SAT-attack iteration using $i$ as a DIP, i.e., 
$$S_i = \{k| f(i)\neq f'(i,k)\}.$$
Finally, let $\mathcal{S}$ the collection of all the sets corresponding to an input pattern, i.e., $\mathcal{S} = \{ S_i | i \in \mathbb{B}^{|I|} \}$. 
Finding the minimum number of DIPs that are enough to prune out all the incorrect keys can then be reduced to finding the minimum number of sets from $\mathcal{S}$ whose union equals $U$, which is a min-set-cover problem. 
\end{proof}

Theorem~\ref{theorem:min_cover} shows that finding the minimum number of DIPs is, in general, a hard problem. We can, however, use greedy algorithms in order to emulate worst-case SAT attacks and provide approximate but conservative estimates for their duration, by  searching and prioritizing the input patterns that can eliminate the largest number of incorrect keys. 
Fig.~\ref{fig:min_dip} shows the largest number of DIPs over all possible values for $h$ returned by the greedy algorithm with different key sizes from 1 to 17. By definition, $t_{SAT}$ should be less than or equal to the results in Fig.~\ref{fig:min_dip}, which exhibits a sub-exponential behavior. 
%
%
\begin{figure}[t]
	\centering
	\includegraphics[width=1\columnwidth]{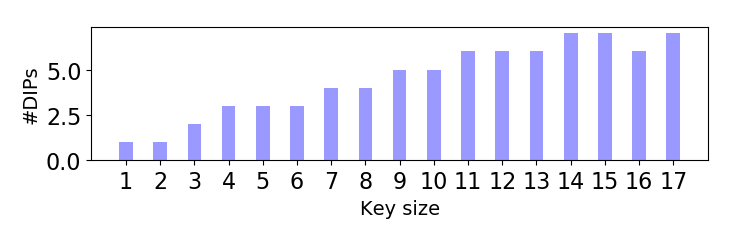}
    \vspace{-20pt}
	\caption{The largest \#DIPs over all possible values for $h$ with different key size returned by the greedy algorithm. }
	\vspace{-15pt}
	\label{fig:min_dip}
\end{figure}
Both the expressions of $E_{FC}$ and $E_{APP}$ shown in Tab.~\ref{tab:chance_cons_formu} can be derived as stated by the following theorems.
\begin{theorem}
For a circuit encrypted with \texttt{SFLL-HD} with Hamming distance parameter $h$, let $l$ and $n$ be the key size and the primary input size, respectively, and $E_{FC}$ the functional corruptibility. We obtain $$E_{FC} = \frac{\binom{l}{h}\left[2^{l} - \binom{l}{h}\right]}{2^{2l-1}}.$$ 
\end{theorem}

\begin{proof}
%
For the circuit in  Figure~\ref{fig:general_locking_decomposition}c, the output is corrupted if and only if $ \mathbbm{1}(HD(i, k^*) = h)\oplus \mathbbm{1}(HD(i, k) = h) = 1$. If $i$ is a protected input pattern, then we have $HD(i, k^*) = h$, while $HD(i, k)$ must be different from $h$ in order to provide a corrupted output. Therefore, the number of key patterns generating an incorrect output for each protected input pattern is $2^{l} - \binom{l}{h}$. Similarly, we can derive the number of key patterns generating an incorrect output for each unprotected input pattern as $\binom{l}{h}$. The total number of protected input patterns and unprotected input patterns can be computed as  $2^{n-l}\binom{l}{h}$ and $2^{n-l}\left[ 2^{l}-\binom{l}{h}\right]$, respectively~\cite{Yasin2017Provably-Secure}.
By summing up all the incorrect output values over all the input and key patterns, we obtain the following expression for the functional corruptibility: 
\begin{equation}
    \begin{aligned}
    E_{FC} &= \frac{
    2^{n-l}
    \left\{\binom{l}{h}\left[2^{l} - \binom{l}{h}\right]+ \left[2^{l} - \binom{l}{h}\right]\binom{l}{h}\right\}
    }
    {2^{n}\cdot 2^{l}}\\
    &=\frac{\binom{l}{h}\left[2^{l} - \binom{l}{h}\right]+ \left[2^{l} - \binom{l}{h}\right]\binom{l}{h}}{2^{2l}}\\
    &=\frac{\binom{l}{h}\left[2^{l} - \binom{l}{h}\right]}{2^{2l-1}} \notag
    \end{aligned}
\end{equation}
\end{proof}

\begin{theorem}
For a circuit encrypted with \texttt{SFLL-HD} with Hamming distance parameter $h$, let $l$ and $n$ be the key size and the primary input size, respectively, and $E_{APP}$ the approximate SAT-attack resilience. We obtain $$E_{APP} = \frac{2\left[\binom{l}{h}-2\binom{l-2}{h-1}\right]}{2^{n}}.$$ 
\end{theorem}

\begin{proof}
We denote by $HD(a, b)$ the Hamming distance between the words (bit strings) $a$ and $b$ and by $Dif(a,b)$ the set of indexes marking the bits that are different in $a$ and $b$. We suppose that the circuit output is corrupted (flipped) for input $i$ and key $k$, and let the HD between $k^*$ (the correct key) and $k$ be $HD(k^*, k) = x$, with $x \in \left[1,l\right], x\in \mathbb{N}$. Then, by recalling the architecture in  Figure~\ref{fig:general_locking_decomposition}c, there can only be an odd number of output inversions and the following equation holds: 
\begin{equation}\notag
    \mathbbm{1}(HD(i, k^*) = h)\oplus \mathbbm{1}(HD(i, k) = h) = 1,
\end{equation}
meaning that one and only one of the two HDs is $h$.  

We first assume that $HD(i, k^*)=h$ and $Dif(i, k^*) \cap Dif(k^*, k) = \emptyset$, i.e., $i$ and $k$ differ from $k^*$ on disjoint sets of indexes. We then conclude that 
$$HD(i, k) = h + x.$$
More generally, if the cardinality of the set $Dif(i, k^*) \cap Dif(k^*, k)$ is  $y$, we obtain 
$$HD(i, k) = h + x-2y.$$
We then consider the following two cases. If 
$x$ is an odd number, then $HD(i, k)$ cannot be equal to $h$ for any $y$, which provides a number of  corrupted outputs equal to $\binom{l}{h}$. On the other hand, if $x$ is even, then $HD(i, k)$ will evaluate to $h$ whenever $x = 2y$. 
Since the number of possible $k$ values satisfying this condition is $\binom{l-x}{h-x/2}\cdot \binom{x}{x/2}$, we obtain a total number of corrupted output values equal to  $\binom{l}{h}-\binom{l-x}{h-x/2}\cdot \binom{x}{x/2}$. Since we are interested in the minimum number of these two cases, we choose $\binom{l}{h}-\binom{l-x}{h-x/2}\cdot \binom{x}{x/2}$ as the minimum number of corrupted outputs when $HD(i, k^*)=h$.


Similar considerations can be applied for the case when $HD(i, k)=h$, which leads to the same conclusion as above. Therefore, the overall number of corrupted outputs is doubled.
The approximate SAT resilience is given by the minimum $E_{FC}$, i.e., 
\begin{equation}\label{eq:AppSAT_level_sim}
    E_{APP} = \min{\left\{\frac{2\left[\binom{l}{h}-\binom{l-x}{h-x/2}\cdot \binom{x}{x/2}\right]}{2^{n}}, \forall x \geq 2, x/2\in \mathbb{N} \right\}}, \notag
\end{equation}
which is achieved for $x=2$, finally leading to  
\begin{equation}\label{eq:AppSAT_level_sim2}
    E_{APP} = \frac{2\left[\binom{l}{h}-2\binom{l-2}{h-1}\right]}{2^{n}} \notag
\end{equation}
\end{proof}

\fakeparagraph{FLL}
\texttt{FLL} aims at creating high $E_{FC}$ with low overhead by appropriately adding key-gates in the circuit, as shown in Fig.~\ref{fig:general_locking_decomposition}d. While the key-gates are not directly inserted at the primary output, their combined effect can still be represented by an appropriate $g$ function producing the same error pattern. While $E_{FC}$ depends on the specific circuit and cannot be computed in closed form, \texttt{FLL} can achieve higher values than  the other three methods, based on empirical results. 
However, it cannot guarantee exponential $t_{SAT}$ with the key size. Moreover, the  
XOR/XNOR-based key-gates may be easy to identify, leading to negligible resilience to removal attacks. 

\begin{figure*}[t]
\vspace{-20pt}
\centering
\subfigure[]{
\centering
\includegraphics[width=0.5\columnwidth]{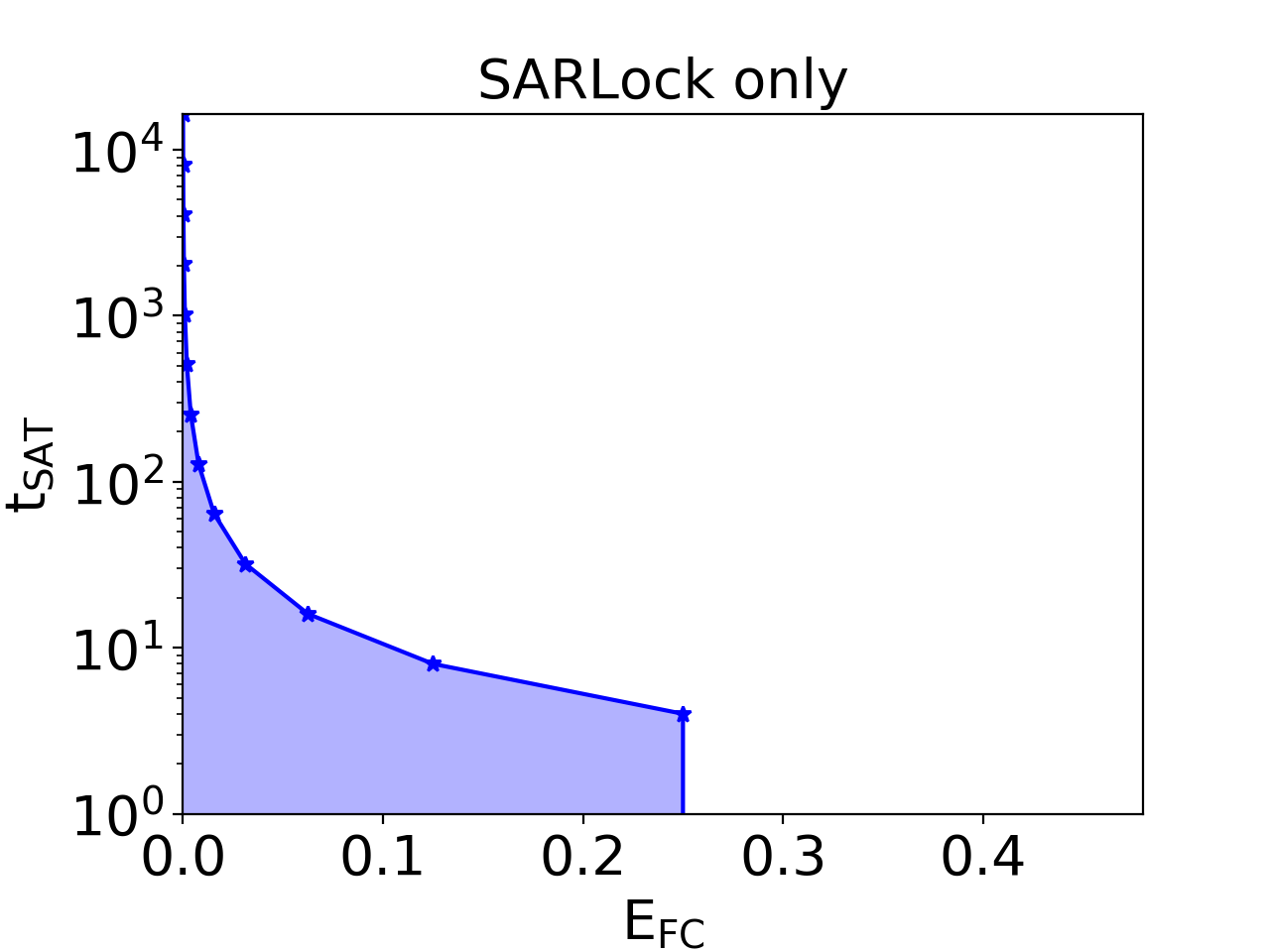}
}%
\subfigure[]{
\centering
\includegraphics[width=0.5\columnwidth]{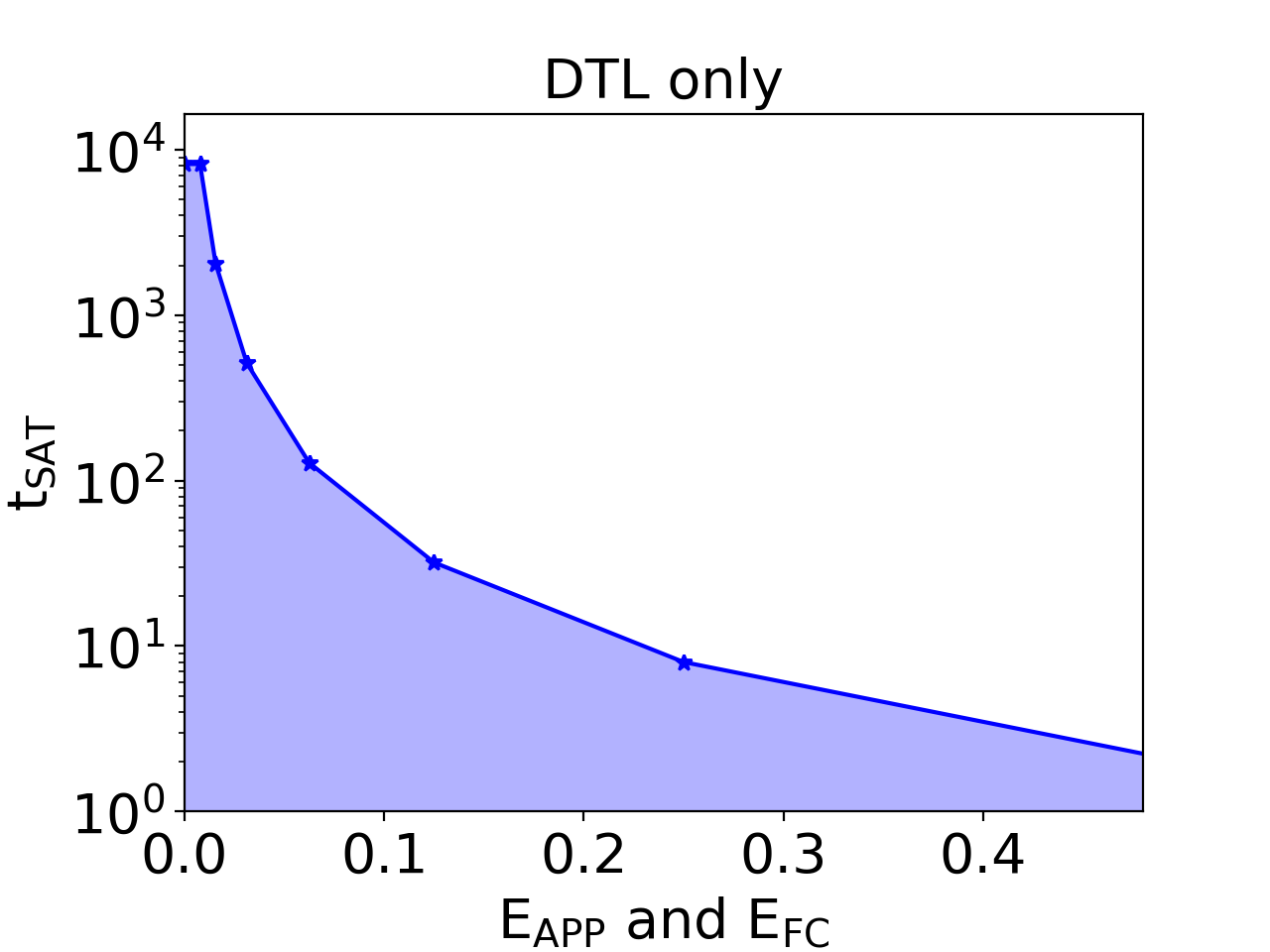}
}%
\subfigure[]{
\centering
\includegraphics[width=0.5\columnwidth]{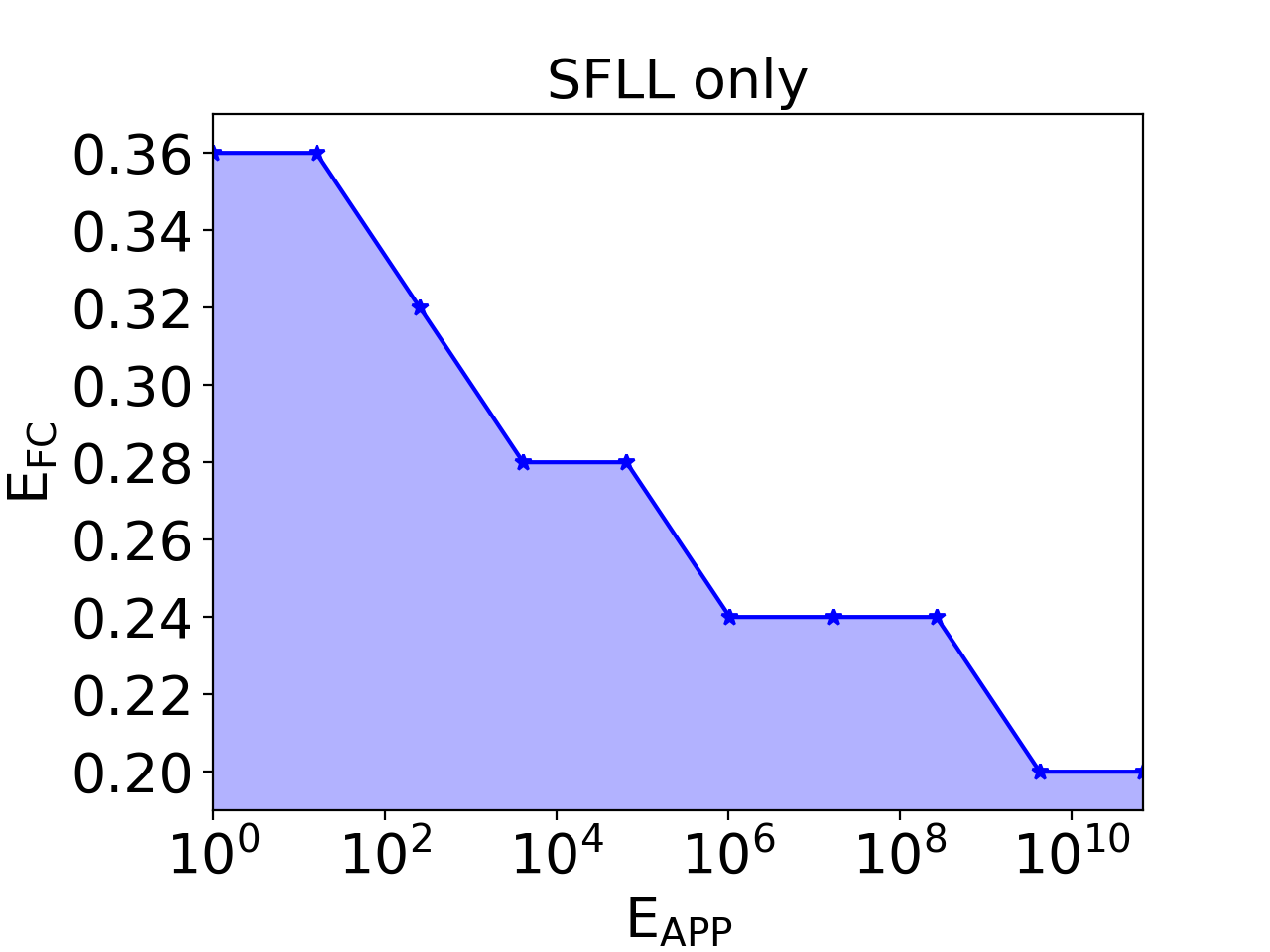}
}%
\subfigure[]{
\centering
\includegraphics[width=0.5\columnwidth]{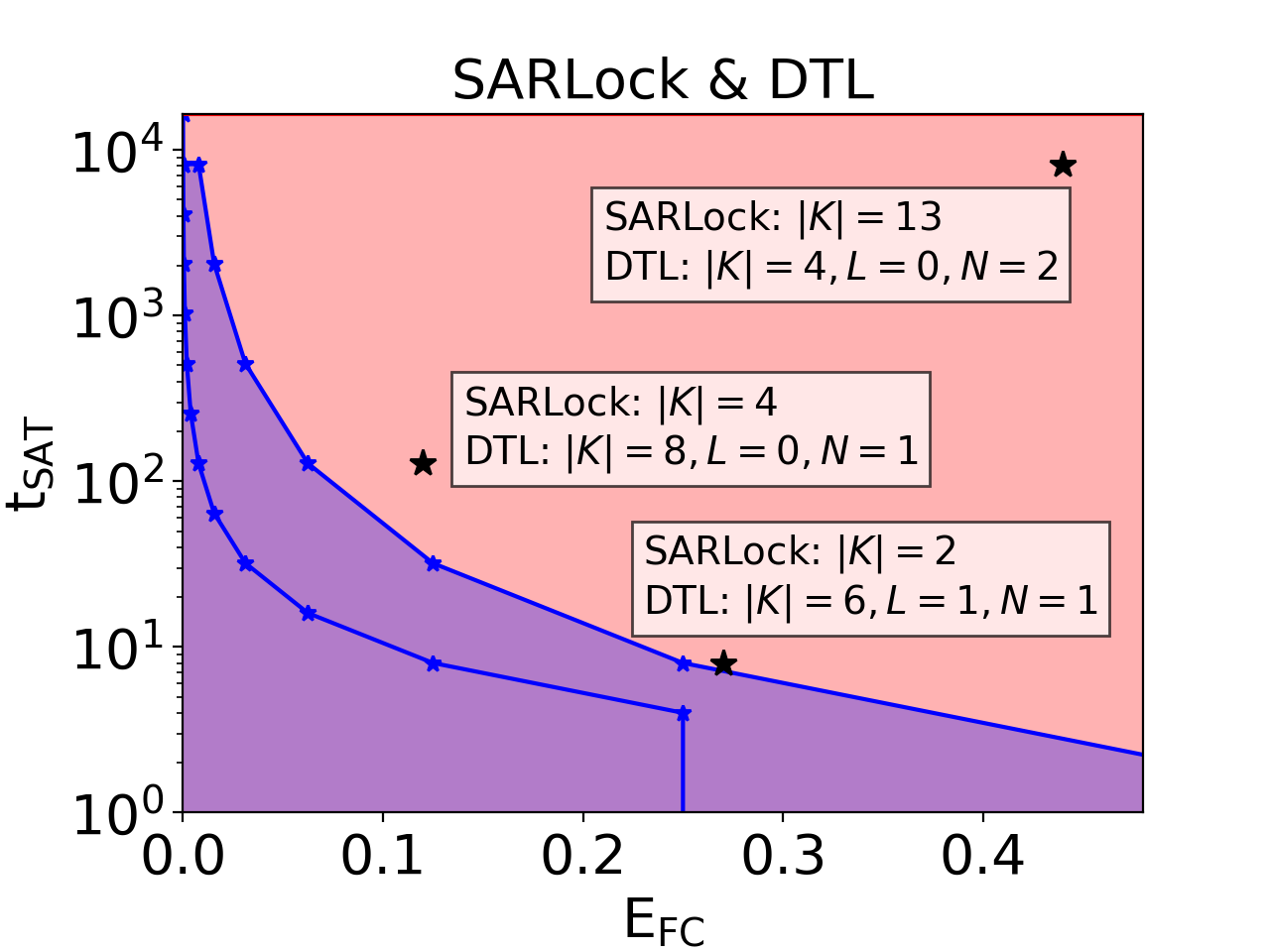}
}%
\centering
\vspace{-10pt}
\caption{Trade-off analysis of (a) \texttt{SARLock}, (b) \texttt{DTL}, (c) \texttt{SFLL}, and (d) \texttt{SARLock} and \texttt{DTL}.
}
\vspace{-10pt}
\label{fig:trade_off}
\end{figure*}

\begin{figure*}[t]
	\centering
	\subfigure[]{
    \centering
    \includegraphics[width=0.5\columnwidth]{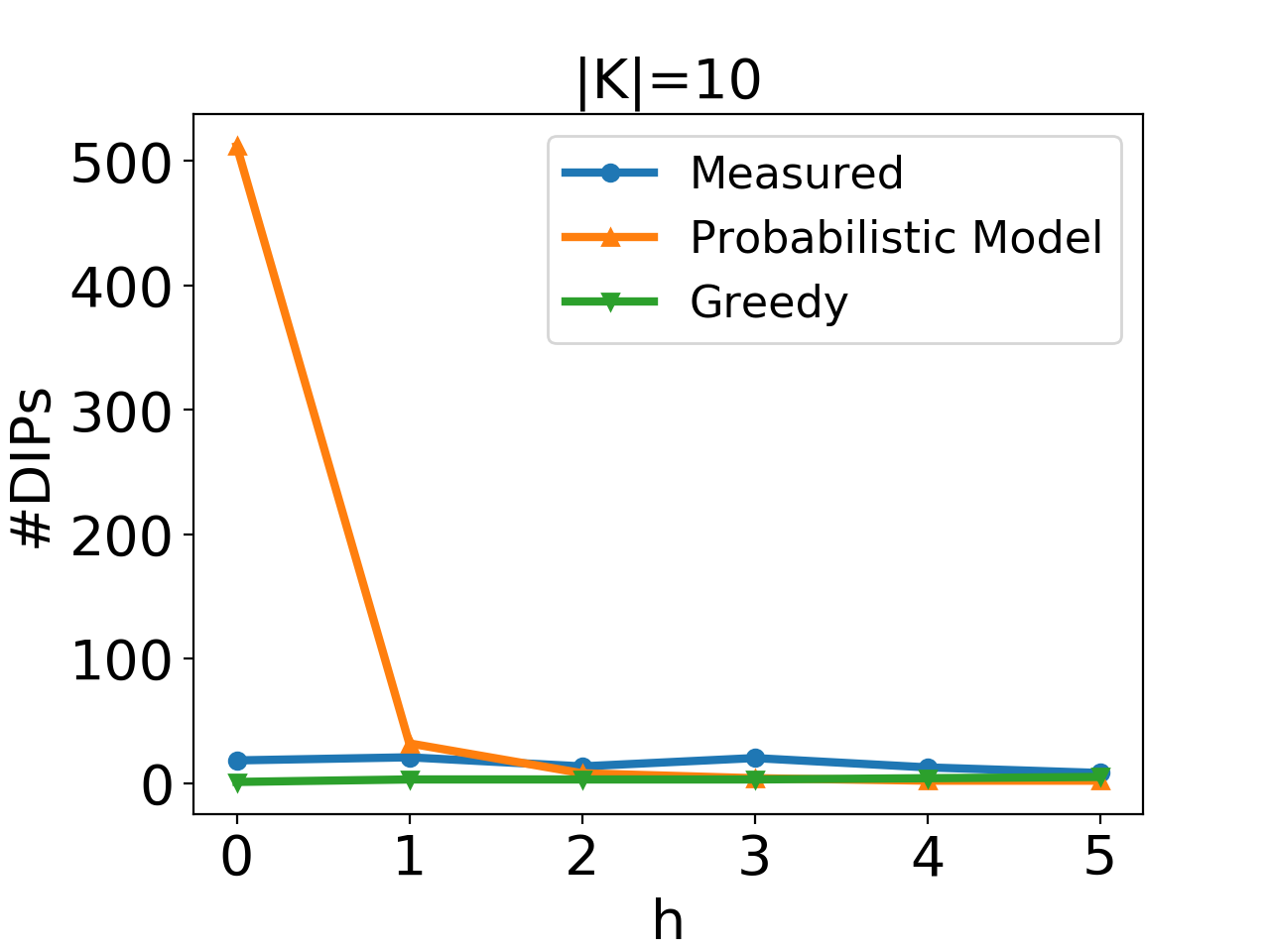}
    }%
    \subfigure[]{
    \centering
    \includegraphics[width=0.5\columnwidth]{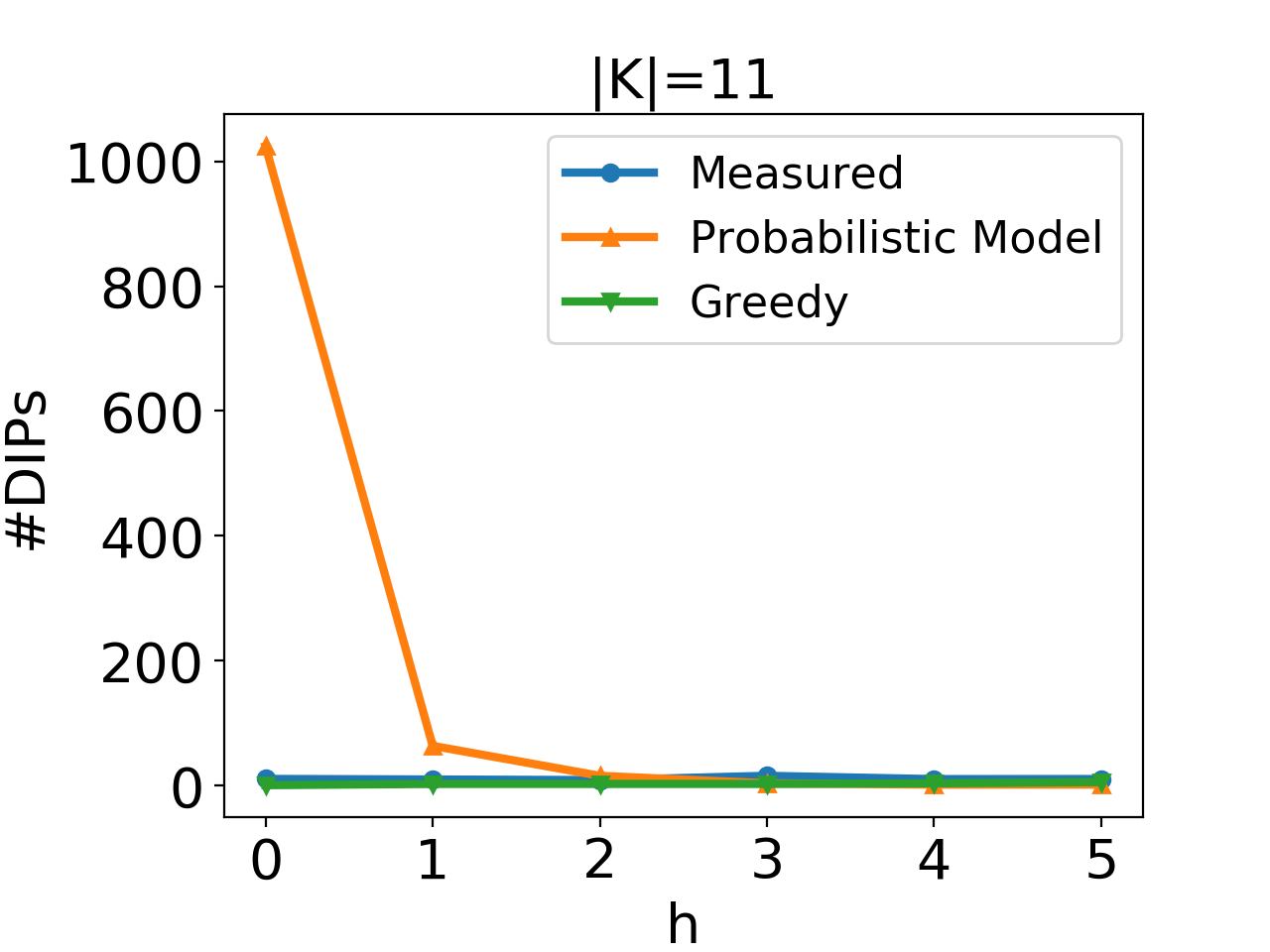}
    }%
    \subfigure[]{
    \centering
    \includegraphics[width=0.5\columnwidth]{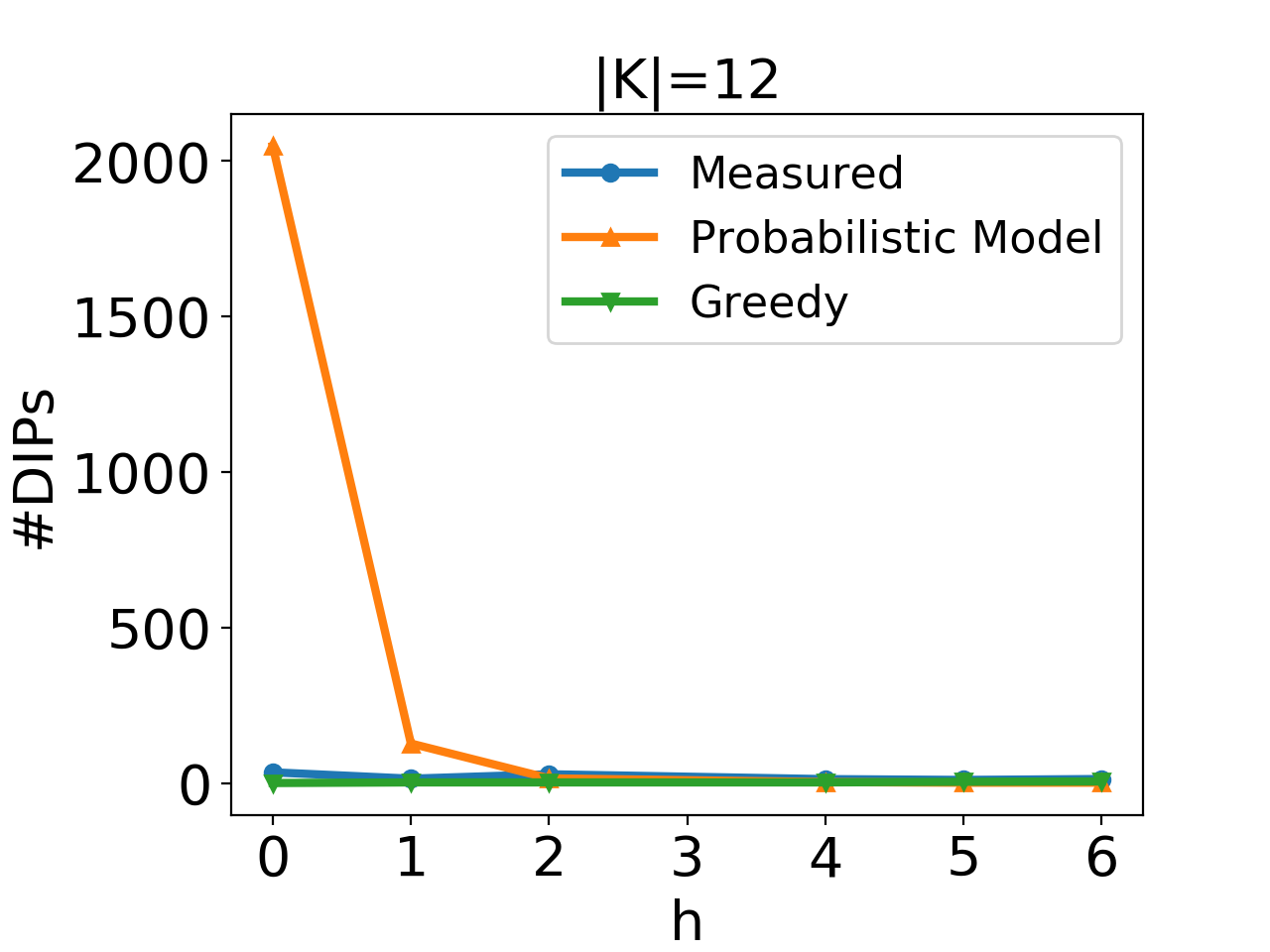}
    }%
    \subfigure[]{
    \centering
    \includegraphics[width=0.5\columnwidth]{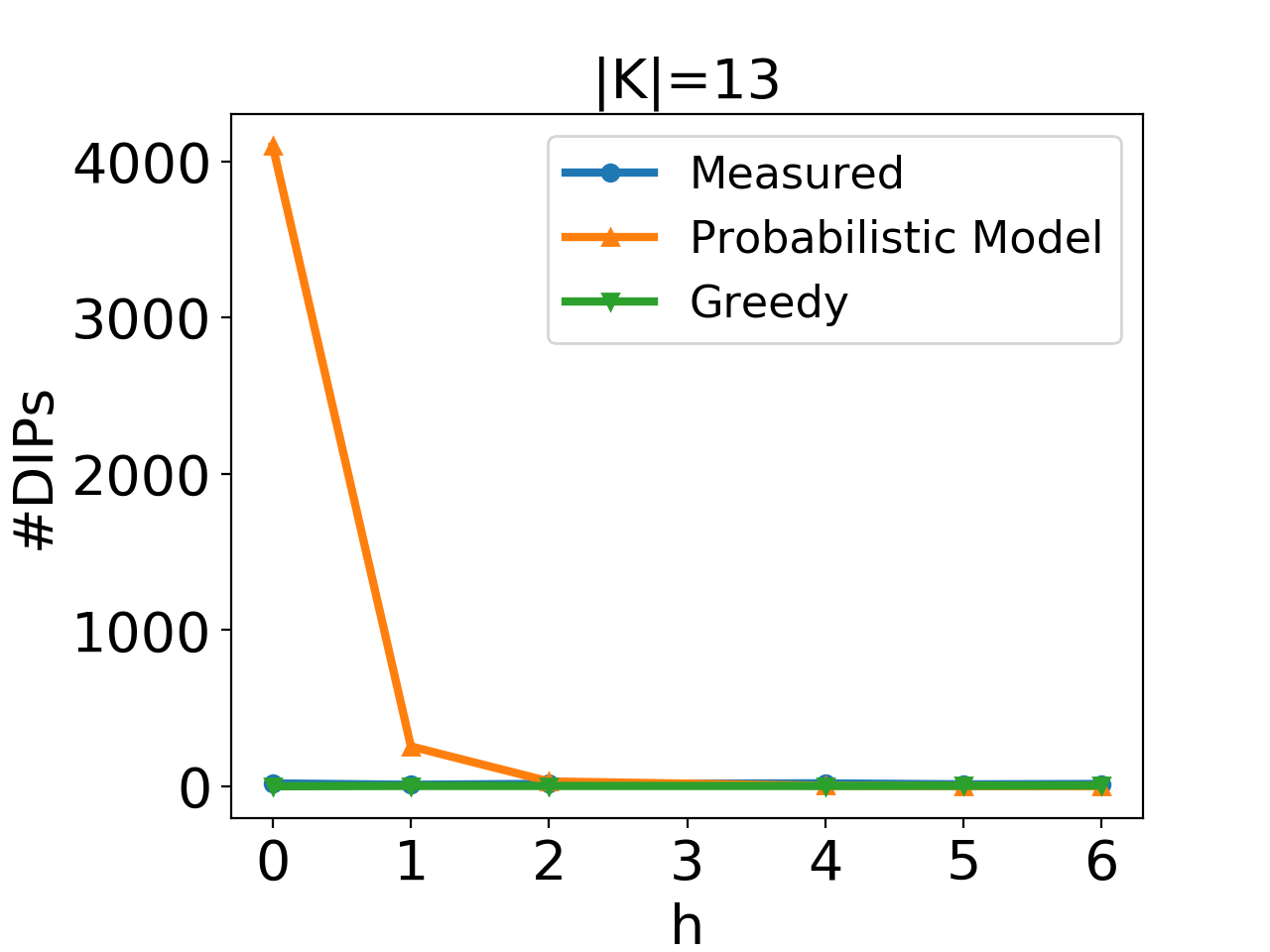}
    }%
	\vspace{-10pt}
	\caption{Average \#DIPs on \texttt{SFLL} encrypted circuits when key size is (a) 10, (b) 11, (c) 12, and (d) 13. 
	}
	\label{fig:sat_attack_sfll}
	\vspace{-10pt}
\end{figure*}

\begin{figure*}[t]
	\centering
	\subfigure[]{
    \centering
    \includegraphics[width=0.5\columnwidth]{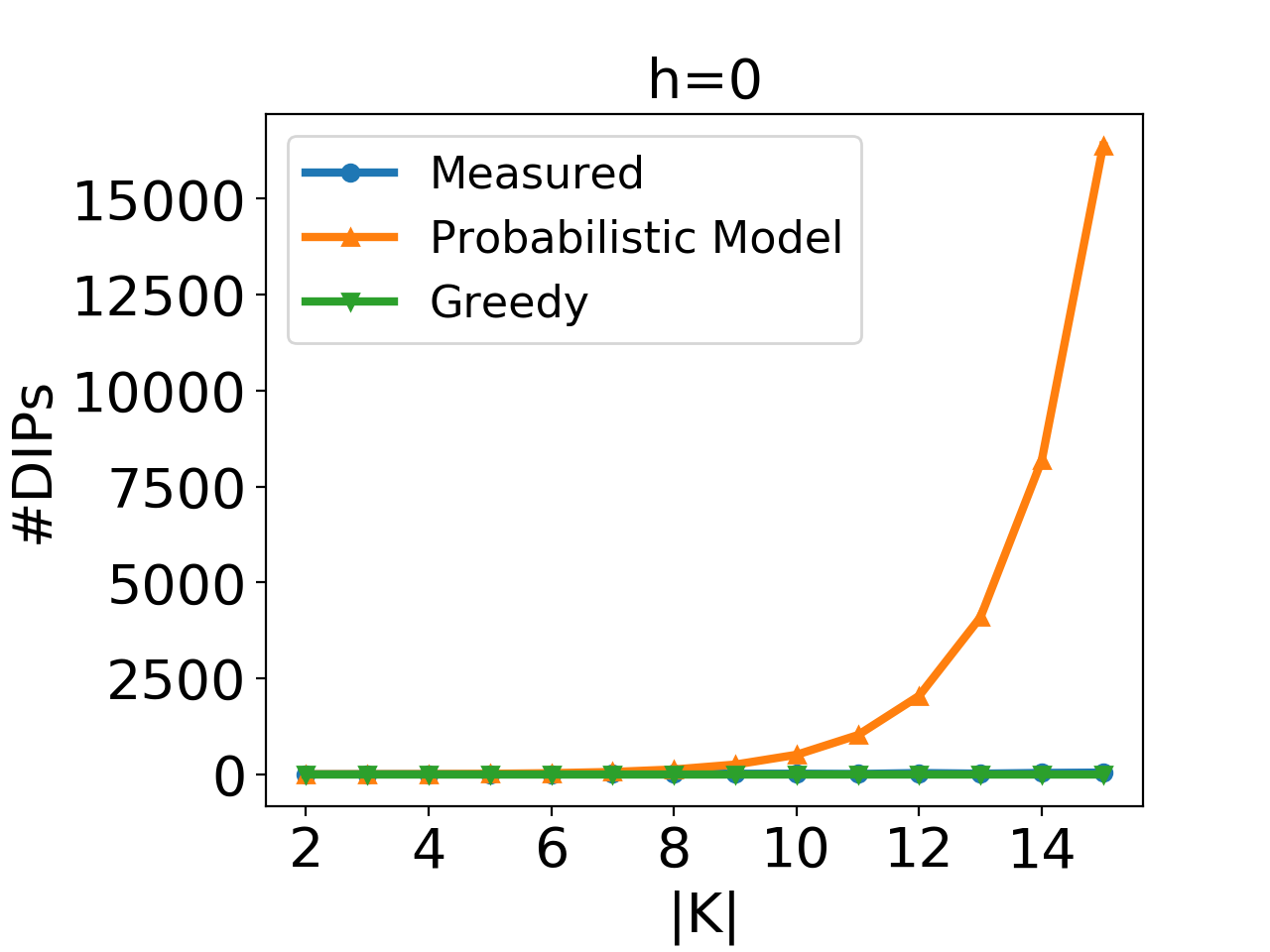}
    }%
    \subfigure[]{
    \centering
    \includegraphics[width=0.5\columnwidth]{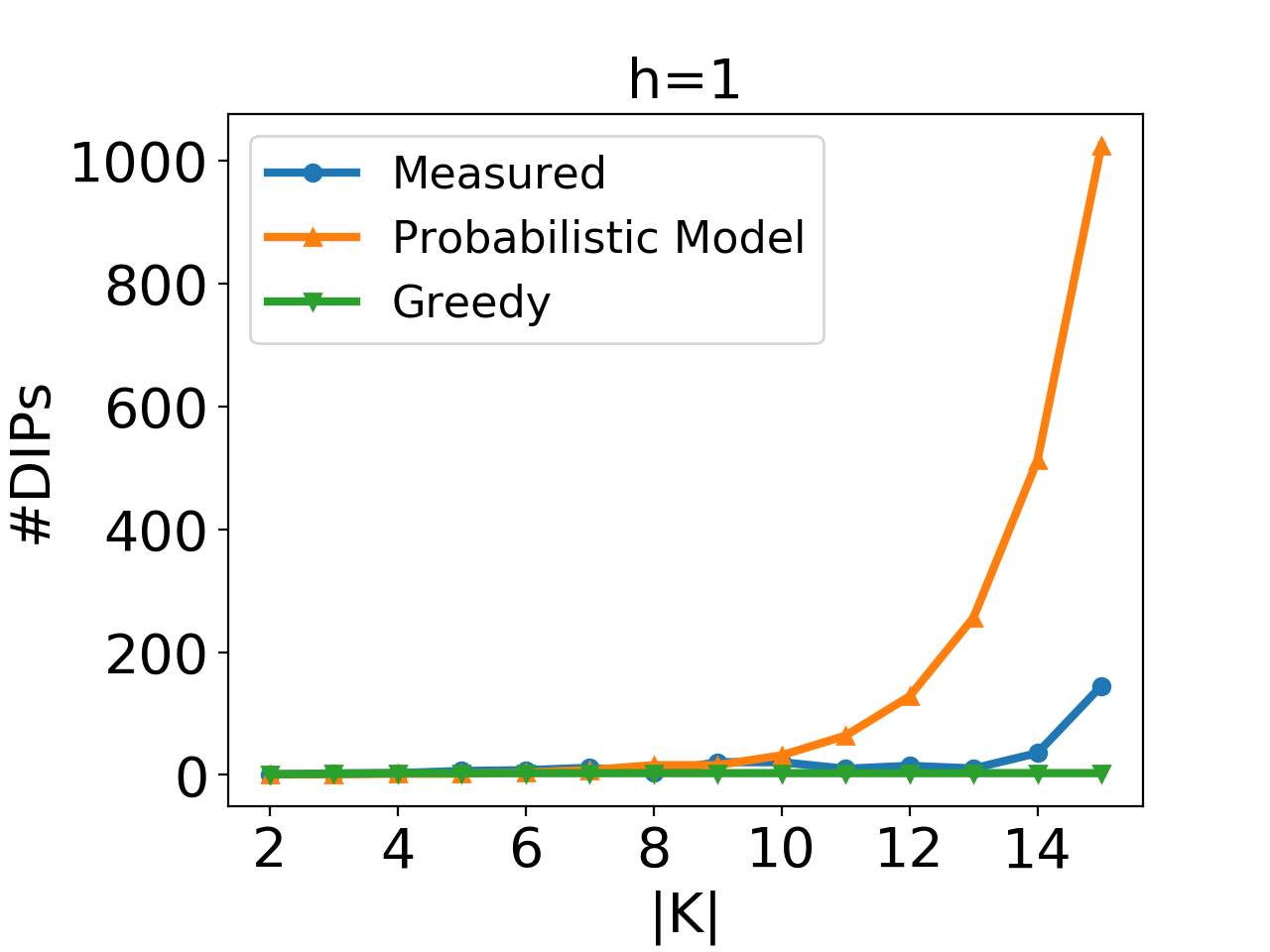}
    }%
    \subfigure[]{
    \centering
    \includegraphics[width=0.5\columnwidth]{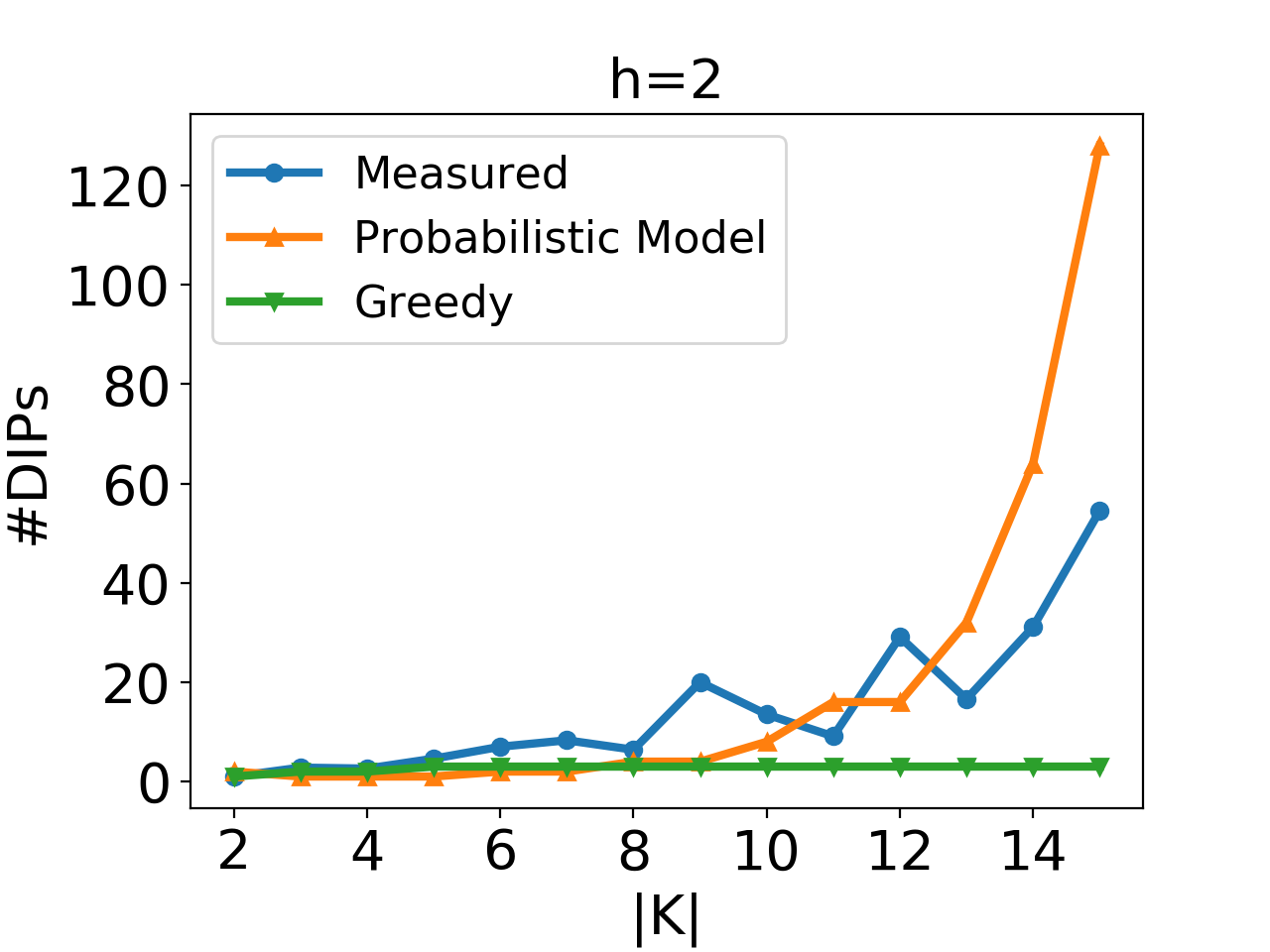}
    }%
    \subfigure[]{
    \centering
    \includegraphics[width=0.5\columnwidth]{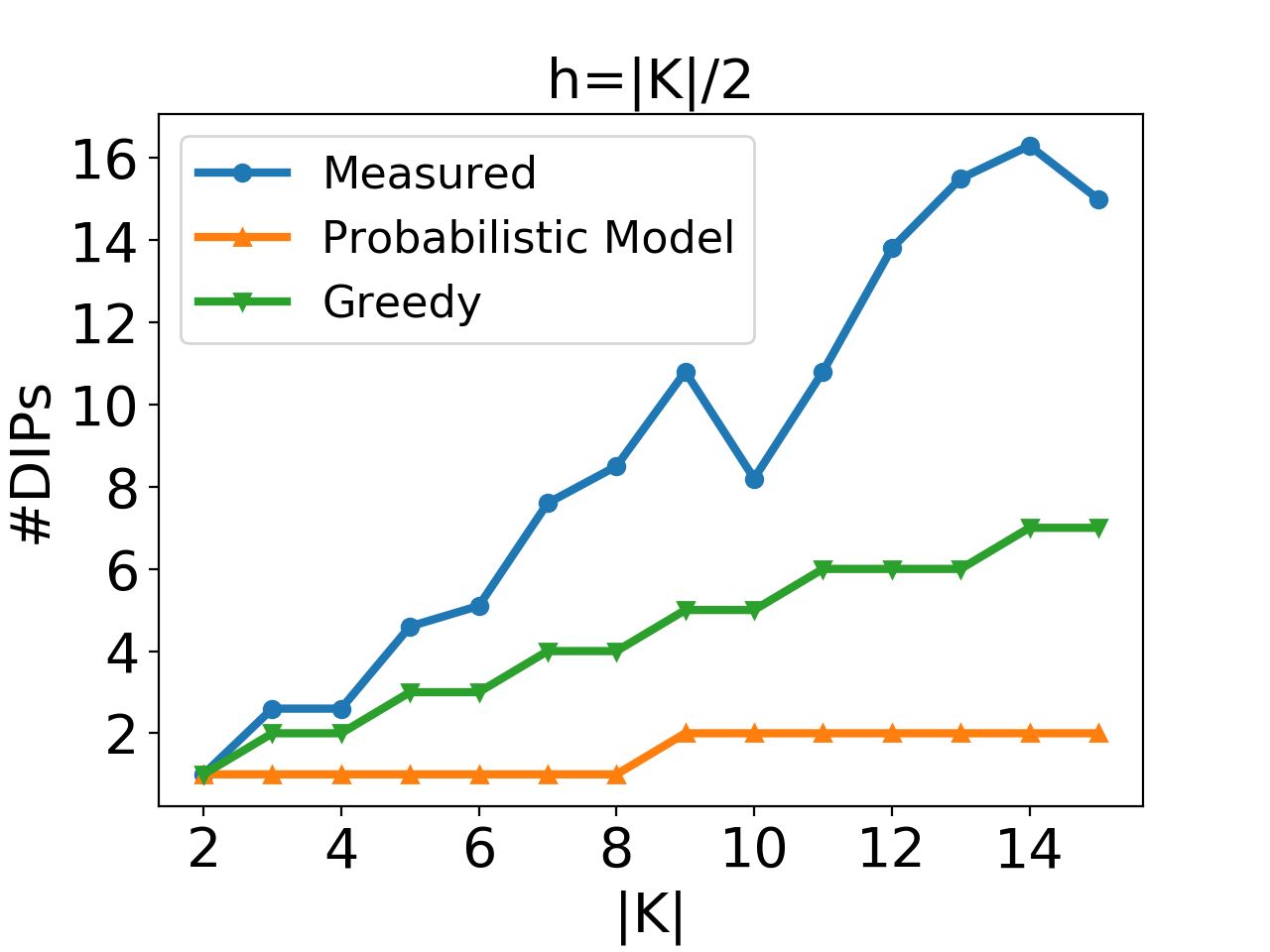}
    }%
	\vspace{-10pt}
	\caption{Average \#DIPs on \texttt{SFLL} encrypted circuits when $h=$ (a) 0, (b) 1, (c) 2, and (d) $\frac{|K|}{2}$. 
	}
	\label{fig:sat_attack_sfll_fix_h}
	\vspace{-10pt}
\end{figure*}

\section{Results and Discussion}\label{sec:experiment}

We evaluated models and metrics on 
a $2.9$-GHz Core-i9 processor with $16$-GB memory.  
%
We first investigated the effectiveness of our models for fast trade-off evaluation on a set of ISCAS benchmark circuits.  
The blue areas in Fig.~\ref{fig:trade_off}a-c pictorially represent, as a continuum, the feasible encryption space for different methods and user requirements. 
%
%
For example, Fig.~\ref{fig:trade_off}a shows that a funtional corruptibility ($E_{FC}$) as high as $0.25$ can still be achieved with \texttt{SARLock}; however, it can only be implemented for very low, and therefore impractical, key sizes.   
Fig.~\ref{fig:trade_off}b highlights the trade-off between SAT attack resilience ($t_{SAT}$) and approximate SAT attack resilience ($E_{APP}$) in \texttt{DTL}. 
As expected, \texttt{DTL} is able to increase $E_{APP}$ and $E_{FC}$ at the cost of decreasing $t_{SAT}$. The highest possible $E_{FC}$ achieved by \texttt{DTL} is higher than that of \texttt{SARLock} in Fig.~\ref{fig:trade_off}a over the same range of keys. 
Finally, Fig.~\ref{fig:trade_off}c exposes a trade-off between $E_{FC}$ and $E_{APP}$ in \texttt{SFLL}. It shows that increasing $E_{FC}$ adversely impacts $E_{APP}$, possibly due to the fact that, as $E_{FC}$ increases, the error distribution is not uniform; while the peak error rate increases, the error can become significantly low for some of the incorrect keys. 

We further implemented all the encryption configurations explored in Fig.~\ref{fig:trade_off}a-c on four ISCAS benchmark circuits, generating $1473$ netlists
in $15$ minutes, to compare the model predictions with the measurements. We used open-source libraries to simulate  
SAT attacks~\cite{Subramanyan2015Evaluating-the-} and report the actual value of $t_{SAT}$. We empirically estimated $E_{FC}$ by averaging the functional corruptibility over  $500$ logic simulations on the encrypted netlists. By using a similar procedure, an empirical estimate for $E_{APP}$ was obtained by taking the average over $500$ logic simulations for each key pattern, and then the minimum corruptibility value over $100$ incorrect key patterns. 
%
%
Fig.~\ref{fig:security_satisfy} reports the results for 
four ISCAS benchmark circuits, showing that the empirical resilience would always exceed the one predicted by our model (blue bar)
for both $t_{SAT}$ and $E_{APP}$.
For $26\%$ of the design (red bar) the empirical $E_{FC}$ proved to be smaller than the predicted one 
by a negligible margin ($4 \times 10^{-3}$), which is within the error affecting our  simulation-based empirical estimates. 

\begin{figure}[t]
	\centering
	\includegraphics[width=0.8\columnwidth]{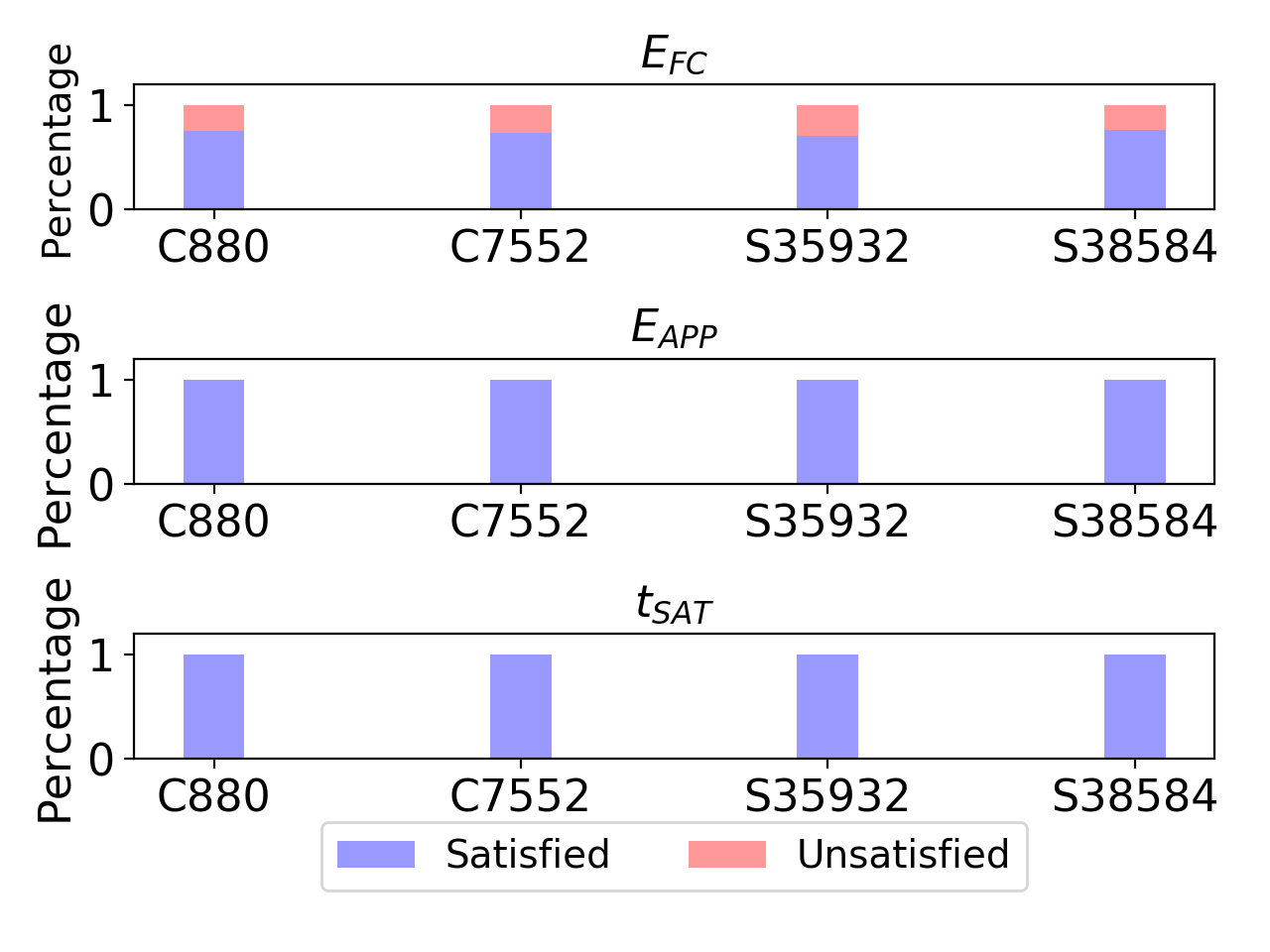}
    \vspace{-10pt}
	\caption{Verification pass rate on different security concerns. 
	}
	\vspace{-15pt}
	\label{fig:security_satisfy}
\end{figure}

To compare our SAT resilience model for \texttt{SFLL} with the measured number of DIPs,  
we simulate SAT attacks on the encrypted netlists of four ISCAS circuits using \texttt{SFLL}-HD. For each combination of key size $|K|$ and HD value $h$, we generate $10$ netlists, by randomly permuting the order of the gates, and compute the average number of DIPs over $10$ SAT attacks. 
As shown in Fig.~\ref{fig:sat_attack_sfll}, when $h$ is close to zero, the predictions of the probabilistic model~\cite{Yasin2017Provably-Secure} exhibit an exponential behavior that significantly differ from the simulated results, and the maximum error can be as high as $20,000\%$. 
Instead, the greedy algorithm predicts the simulated number of DIPs for all key sizes and $h$ values with relative errors that are less than or equal to $97\%$, two orders of magnitude smaller than previous approaches. For $h\!>\!0$, the average prediction error of the greedy algorithm becomes twice as small as the one of the probabilistic model. 

Fig.~\ref{fig:sat_attack_sfll_fix_h}a-d show the relation between the number of DIPs and the key size $|K|$ when $h=$ 0, 1, 2, and $\frac{|K|}{2}$, respectively. In Fig.~\ref{fig:sat_attack_sfll_fix_h}a and Fig.~\ref{fig:sat_attack_sfll_fix_h}b, when $h$ is close to zero, the probabilistic model offers an estimate of the number of DIPs which deviates from the simulated result, while the greedy algorithm always returns a closer, more conservative  prediction. 
In Fig.~\ref{fig:sat_attack_sfll_fix_h}c, the greedy algorithm shows better accuracy than the probabilistic model when $|K|\leq 7$. For the other key sizes, the probabilistic model outperforms the greedy algorithm. However, the prediction provided by the probabilistic model grows faster than in  simulation and tend to overestimate the number of DIPs when $|K| \geq 13$. 
In Fig.~\ref{fig:sat_attack_sfll_fix_h}d, when $h=\frac{|K|}{2}$, the prediction from the probabilistic model becomes more conservative. Conversely, the greedy algorithm estimates an average error twice as small as the probabilistic model. 

Overall, the aforementioned results reveal the inherent difficulty of achieving high security levels against multiple threats using a single technique. However, this challenge may be addressed by combining multiple techniques. To test this hypothesis,
we encrypted the ISCAS circuit C880 with both \texttt{SARLock} and \texttt{DTL}, by using a logic OR gate to combine their output 
(flip) signals. We then combined the output of the OR gate with the output of the original circuit via a XOR gate. 
Fig.~\ref{fig:trade_off}d shows that the compound strategy significantly alleviates the trade-offs posed by \texttt{SARLock} alone, making it possible to achieve both high functional corruptibility and SAT-attack resilience. 
For example, the topmost configuration in Fig.~\ref{fig:trade_off}d achieves $t_{SAT}\geq 2^{14}$ and $E_{FC} \geq 0.48$, which cannot be obtained with \texttt{SARLock} or \texttt{DTL} alone. The compound scheme, where the \texttt{SARLock} block has key size 13 and the \texttt{DTL} block has key size 4, with two AND gates being replaced by one XOR gate in the first layer, is able to provide both high $t_{SAT}$ and $E_{FC}$. Simulation results are, again, in  agreement with our model predictions,  
showing that our models can indeed be used to capture the performance of compound encryption schemes. 

\section{Conclusions}\label{sec:conclusion}

Simulation results show the effectiveness of the proposed models and metrics for fast and accurate evaluation of the design trade-offs as well as the exploration of compound logic encryption strategies, which may be required for protecting against different threats with small overhead. Future extension of this work include the  incorporation of overhead models as well as support for structural and learning-based attacks. We plan to also investigate the extension of our framework to sequential logic encryption methods. 
Finally, we plan to further develop an automated design and verification environment~\cite{vivek2019system} leveraging  our models and methods to perform design space exploration and inform system-level design decisions across multiple encryption schemes. 

\section*{Acknowledgments}\label{sec:ack}
This work was partially sponsored by the Air Force Research Laboratory (AFRL) and the Defense Advanced Research Projects Agency (DARPA) under agreement number FA8560-18-1-7817. 

\bibliographystyle{ieeetr}
\bibliography{archEx,mirage_ref} 

\begin{thebibliography}{10}

\bibitem{hu2019models}
Y.~Hu, V.~V. Menon, A.~Schmidt, J.~Monson, M.~French, and P.~Nuzzo,
  ``Security-driven metrics and models for efficient evaluation of logic
  encryption schemes,'' in {\em ACM-IEEE Int. Conf. Formal Methods and Models
  for System Design (MEMOCODE)}, 2019.

\bibitem{Roy2010Ending-piracy-o}
J.~A. Roy, F.~Koushanfar, and I.~L. Markov, ``{EPIC}: Ending piracy of
  integrated circuits,'' in {\em Proc. Conf. Design, automation and test in
  Europe (DATE)}, pp.~1069--1074, 2008.

\bibitem{Rajendran2015Fault-Analysis-}
J.~Rajendran, H.~Zhang, C.~Zhang, G.~S. Rose, Y.~Pino, O.~Sinanoglu, and
  R.~Karri, ``Fault analysis-based logic encryption,'' {\em IEEE Trans.
  Computers}, vol.~64, no.~2, pp.~410--424, 2013.

\bibitem{Yasin2016SARLock:-SAT-at}
M.~Yasin, B.~Mazumdar, J.~Rajendran, and O.~Sinanoglu, ``{SARLock}: {SAT}
  attack resistant logic locking,'' in {\em IEEE Int. Symp. Hardware Oriented
  Security and Trust (HOST)}, pp.~236--241, 2016.

\bibitem{shamsi2019approximation}
K.~Shamsi, T.~Meade, M.~Li, D.~Z. Pan, and Y.~Jin, ``On the approximation
  resiliency of logic locking and {IC} camouflaging schemes,'' {\em IEEE Trans.
  Information Forensics and Security}, vol.~14, no.~2, pp.~347--359, 2019.

\bibitem{tehranipoor2010survey}
M.~Tehranipoor and F.~Koushanfar, ``A survey of hardware trojan taxonomy and
  detection,'' {\em IEEE Design \& Test of Computers}, vol.~27, no.~1,
  pp.~10--25, 2010.

\bibitem{Subramanyan2015Evaluating-the-}
P.~Subramanyan, S.~Ray, and S.~Malik, ``Evaluating the security of logic
  encryption algorithms,'' in {\em IEEE Int. Symp. Hardware Oriented Security
  and Trust (HOST)}, pp.~137--143, 2015.

\bibitem{Yasin2017Provably-Secure}
M.~Yasin, A.~Sengupta, M.~T. Nabeel, M.~Ashraf, J.~Rajendran, and O.~Sinanoglu,
  ``Provably-secure logic locking: From theory to practice,'' in {\em Proc. ACM
  SIGSAC Conf. Computer and Communications Security}, pp.~1601--1618, 2017.

\bibitem{zhou2017humble}
H.~Zhou, ``A humble theory and application for logic encryption.,'' {\em IACR
  Cryptology ePrint Archive}, vol.~2017, p.~696, 2017.

\bibitem{valiant1984theory}
L.~G. Valiant, ``A theory of the learnable,'' in {\em Proc. ACM Symp. Theory of
  Computing}, pp.~436--445, 1984.

\bibitem{shamsi2019locking}
K.~Shamsi, D.~Z. Pan, and Y.~Jin, ``On the impossibility of
  approximation-resilient circuit locking,'' in {\em IEEE Int. Symp. Hardware
  Oriented Security and Trust (HOST)}, pp.~161--170, 2019.

\bibitem{barak2001possibility}
B.~Barak, O.~Goldreich, R.~Impagliazzo, S.~Rudich, A.~Sahai, S.~Vadhan, and
  K.~Yang, ``On the (im)possibility of obfuscating programs,'' in {\em Int.
  Cryptology Conf.}, pp.~1--18, Springer, 2001.

\bibitem{goldwasser2007best}
S.~Goldwasser and G.~N. Rothblum, ``On best-possible obfuscation,'' in {\em
  Theory of Cryptography Conf.}, pp.~194--213, Springer, 2007.

\bibitem{yasin2017ttlock}
M.~Yasin, B.~Mazumdar, J.~Rajendran, and O.~Sinanoglu, ``{TTLock}: Tenacious
  and traceless logic locking,'' in {\em IEEE Int. Symp. Hardware Oriented
  Security and Trust (HOST)}, pp.~166--166, 2017.

\bibitem{shamsi2017appsat}
K.~Shamsi, M.~Li, T.~Meade, Z.~Zhao, D.~Z. Pan, and Y.~Jin, ``{AppSAT}:
  Approximately deobfuscating integrated circuits,'' in {\em IEEE Int. Symp.
  Hardware Oriented Security and Trust (HOST)}, pp.~95--100, 2017.

\bibitem{shen2017double}
Y.~Shen and H.~Zhou, ``Double {DIP}: Re-evaluating security of logic encryption
  algorithms,'' in {\em ACM Proc. Great Lakes Symp. VLSI}, pp.~179--184, 2017.

\bibitem{chakraborty2018sail}
P.~Chakraborty, J.~Cruz, and S.~Bhunia, ``{SAIL}: Machine learning guided
  structural analysis attack on hardware obfuscation,'' in {\em IEEE Asian
  Hardware Oriented Security and Trust Symp. (AsianHOST)}, pp.~56--61, 2018.

\bibitem{Xie2018Anti-SAT:-Mitig}
Y.~Xie and A.~Srivastava, ``Anti-{SAT}: Mitigating {SAT} attack on logic
  locking,'' {\em IEEE Trans. Computer-Aided Design of Integrated Circuits and
  Systems}, vol.~38, no.~2, pp.~199--207, 2018.

\bibitem{korte2012combinatorial}
B.~Korte and J.~Vygen, {\em Combinatorial Optimization}, vol.~2.
\newblock Springer, 2012.

\bibitem{vivek2019system}
V.~V. Menon, G.~Kolhe, A.~Schmidt, J.~Monson, M.~French, Y.~Hu, P.~A. Beerel,
  and P.~Nuzzo, ``System-level framework for logic obfuscation with quantified
  metrics for evaluation,'' in {\em IEEE Secure Development Conference
  (SecDev)}, 2019.

\end{thebibliography}
\end{document}